\documentclass[11pt,letter]{article}
\usepackage[margin=1in]{geometry}
\usepackage{amssymb,amsthm,amsmath,mathtools}
\usepackage{thmtools,thm-restate}
\usepackage[square]{natbib}
\usepackage{pifont}
\usepackage{palatino}

\usepackage[usenames,svgnames,xcdraw,table]{xcolor}
\definecolor{DarkGreen}{rgb}{0.1,0.5,0.1}

\usepackage[ruled, lined, noend, linesnumbered, boxed]{algorithm2e}

\SetAlFnt{\small}
\SetAlCapFnt{\small}
\SetAlCapNameFnt{\small}
\SetAlCapHSkip{0pt}
\IncMargin{-\parindent}
\SetKwFor{When}{when}{do}{end}
\SetKw{Break}{break}

\usepackage{comment}
\usepackage{amsthm}
\usepackage{thmtools,thm-restate}
\usepackage{float}
\usepackage{graphicx}
\usepackage{colortbl}
\usepackage{tikz-cd}
\usepackage{enumerate}
\usepackage{hyperref}
\hypersetup{
     colorlinks   = true,
     linkcolor    = red, 
     urlcolor     = teal, 
	 citecolor    = blue 
}
\usepackage[capitalise,noabbrev]{cleveref}
\Crefname{property}{Property}{Properties}
\Crefname{example}{Example}{Examples}
\Crefname{table}{Table}{Tables}

\usepackage{mathtools}
\usepackage{booktabs}
\usepackage{multirow,bigdelim}
\usepackage{blindtext}
\usepackage{etoolbox}
\usepackage{xfrac}  
\usepackage{todonotes}
\usepackage{thm-restate}

\usepackage{tikz, subcaption}
\usetikzlibrary{patterns,arrows}
\usetikzlibrary{shapes.geometric}
\usetikzlibrary{decorations.pathreplacing,calligraphy}

\usepackage{subfiles}
\usepackage{authblk}

\usepackage[capitalise,noabbrev]{cleveref}
\Crefname{remark}{Remark}{Remarks}
\Crefname{rmk}{Remark}{Remarks}
\Crefname{dfn}{Definition}{Definitions}
\Crefname{thm}{Theorem}{Theorems}
\Crefname{cor}{Corollary}{Corollaries}
\Crefname{lem}{Lemma}{Lemmas}
\Crefname{examplex}{Example}{Examples}
\Crefname{prop}{Proposition}{Propositions}

\theoremstyle{plain}
\newtheorem{theorem}{Theorem}
\newtheorem{proposition}{Proposition}

\newtheorem{lemma}{Lemma}

\theoremstyle{definition}
\newtheorem{definition}{Definition}
\newtheorem{example}{Example}
\newtheorem{remark}{Remark}

\renewcommand{\tilde}{\widetilde}
\renewcommand{\bar}{\overline}
\newcommand*\dif{\mathop{}\!\mathrm{d}}
\newcommand{\E}{\mathbb{E}}
\renewcommand{\le}{\leqslant}
\renewcommand{\leq}{\leqslant}
\renewcommand{\ge}{\geqslant}
\renewcommand{\geq}{\geqslant}
\newcommand{\set}[1]{\left\{{#1}\right\}}

\newcommand{\bfy}{\boldsymbol{y}}

\newcommand{\prop}{\mathsf{prop}}
\newcommand{\mms}{\mathsf{mms}}
\newcommand{\usw}{\mathsf{usw}}
\newcommand\e{1-\sfrac{1}{e}}
\newcommand{\half}{\sfrac{1}{2}}

\newcommand{\CEF}[1]{\ifstrempty{#1}{\textrm{\textup{CEF1}}}{$#1$-\textrm{\textup{CEF1{}}}}}
\newcommand{\cef}[1]{\ifstrempty{#1}{\textrm{\textup{CEF}}}{$#1$-\textrm{\textup{CEF{}}}}}

\newcommand{\PEF}[1]{\ifstrempty{#1}{\textrm{\textup{PEF1}}}{$#1$-\textrm{\textup{PEF1{}}}}}
\newcommand{\pef}[1]{\ifstrempty{#1}{\textrm{\textup{PEF}}}{$#1$-\textrm{\textup{PEF{}}}}}

\newcommand{\CPROP}[1]{\ifstrempty{#1}{\textrm{\textup{CPROP}}}{$#1$-\textrm{\textup{CPROP{}}}}}
\newcommand{\CMMS}[1]{\ifstrempty{#1}{\textrm{\textup{CMMS}}}{$#1$-\textrm{\textup{CMMS{}}}}}
\newcommand{\USW}[1]{\ifstrempty{#1}{\textrm{\textup{USW}}}{$#1$-\textrm{\textup{USW{}}}}}
\newcommand{\NW}[1]{\textrm{\textup{NW}}}

\DeclareMathOperator*{\argmin}{\arg\min}

\newcommand{\algdetindiv}{\textsc{Match-and-Shift}\xspace}
\newcommand{\algdetdiv}{\textsc{Equal-Filling}\xspace}
\newcommand{\algrandindiv}{\textsc{Equal-Filling-OCS}\xspace}
\newcommand{\algdiscussion}{\textsc{Equal-Ranking}\xspace}
\newcommand{\ranking}{\textsc{Ranking}\xspace}

\begin{document}

\title{Class Fairness in Online Matching}
\date{}

\author[1]{Hadi Hosseini}
\author[2]{Zhiyi Huang}
\author[3]{Ayumi Igarashi}
\author[4]{Nisarg Shah}

\affil[1]{Pennsylvania State University\\
	{\small\url{hadi@psu.edu}}}
\affil[2]{University of Hong Kong\\
	{\small\url{zhiyi@cs.hku.hk}}}
\affil[3]{National Institute of Informatics\\
	{\small\url{ayumi_igarashi@nii.ac.jp}}}
\affil[4]{University of Toronto\\
	{\small\url{nisarg@cs.toronto.edu}}}

\maketitle

\begin{abstract}
In the classical version of online bipartite matching, there is a given set of offline vertices (aka agents) and another set of vertices (aka items) that arrive online. When each item arrives, its incident edges---the agents who like the item---are revealed and the algorithm must irrevocably match the item to such agents. 

We initiate the study of \emph{class fairness} in this setting, where agents are partitioned into a set of classes and the matching is required to be fair with respect to the classes. We adopt popular fairness notions from the fair division literature such as envy-freeness (up to one item), proportionality, and maximin share fairness to our setting. Our class versions of these notions demand that all classes, regardless of their sizes, receive a fair treatment. We study deterministic and randomized algorithms for matching indivisible items (leading to integral matchings) and for matching divisible items (leading to fractional matchings). 

We design and analyze three novel algorithms. For matching indivisible items, we propose an adaptive-priority-based algorithm, \algdetindiv, prove that it achieves $\half$-approximation of both class envy-freeness up to one item and class maximin share fairness, and show that each guarantee is tight. For matching divisible items, we design a water-filling-based algorithm, \algdetdiv, that achieves $(\e)$-approximation of class envy-freeness and class proportionality; we prove $\e$ to be tight for class proportionality and establish a $\sfrac{3}{4}$ upper bound on class envy-freeness. Finally, we build upon \algdetdiv to design a randomized algorithm for matching indivisible items, \algrandindiv, which achieves $0.593$-approximation of class proportionality. The algorithm and its analysis crucially leverage the recently introduced technique of online correlated selection (OCS)~\citep{fahrbach2020edge}.
\end{abstract}

\section{Introduction}

The one-sided matching problem is a fundamental subject within economics and computation that deals with the matching of a set of items to a set of agents. Its primary objective is to ensure desirable normative properties such as \textit{economic efficiency} and \textit{fairness}. 
The advent of Internet economics along with the introduction of novel marketplaces has posed new challenges in designing desirable solutions for which, as noted by \citet{moulin2019fair}, ``\textit{we need division rules that are both transparent and agreeable, in other words, fair.}''
A wide array of these applications are inherently online, that is, items (or goods) arrive in an online fashion, and need to be matched immediately and irrevocably to the participating agents: consider the examples of allocating advertisement slots to Internet advertisers \citep{mehta2007adwords}, assigning packets to output ports in switch routing \citep{azar2005management}, distributing food donations among nonprofit charitable organizations \citep{lee2019webuildai}, and matching riders to drivers in ridesharing platforms \citep{banerjee2019ride}.
Over the past few decades, a large body of literature---within the field of online algorithm design---is devoted to the study of \textit{online bipartite matching} problems. Their primary goal is to satisfy some notion of economic efficiency---e.g. maximizing the size of the final matching---with no knowledge of which items will arrive in the future and in what order. Algorithms designed for this problem are judged by their \emph{competitive ratio}, which is the worst-case approximation ratio of the size of the matching produced to the maximum possible size in hindsight. It is well known that the best deterministic algorithm can only achieve a $\half$-approximation of this efficiency goal, e.g., by using a greedy algorithm to get a maximal matching. Notably, the seminal work of \citet{karp1990optimal} provides a randomized algorithm called \ranking with the best possible $(\e)$-approximation.

While the literature offers online algorithms with optimal efficiency guarantees, little work has been done in ensuring that these algorithms treat agents, or rather, classes of agents \emph{fairly}. 
Consider the example of a food bank that wishes to distribute the donated items among nonprofit organizations and homeless shelters. The perishable food items donated to the food bank must be assigned upon their arrival. How should an online matching algorithm distribute these donations to the nonprofits and shelters in such a manner that the communities they serve are treated equitably?

\paragraph{Class fairness.} We initiate the study of class fairness in online matching, where a set of items arriving online must be assigned to agents, who are partitioned into known classes, with the goal of achieving fairness among classes. Agents either like an item (value $1$) or don't like it (value $0$). We adopt classical notions from the fair division literature that typically apply to individual agents---such as \textit{envy-freeness} (EF), \textit{proportionality} (Prop), and \textit{maximin share guarantee} (MMS)---to classes of agents. Our extensions ensure that different classes are treated equally, 
regardless of their sizes (e.g., in the food bank example above, different communities are treated equally, even if some have many more organizations serving them). 

Consider, for example, the appealing notion of envy-freeness, which, when applied to individual agents, demands that no agent envy the resources given to another agent. 
When applied to classes, our \emph{class envy-freeness} (\cef{}) notion requires that no class of agents be able to increase their total value by taking the items matched to another class, \emph{even if} it assigns these items optimally among its members. 
With indivisible items (which must be assigned entirely to a single agent), a class envy-free matching may not always exist: consider a single item to be divided between two classes with one agent each liking the item. In the standard fair division model, this impossibility has motivated relaxations such as envy-freeness up to one item (EF1), which can be guaranteed~\cite{lipton2004approximately}. When applied to classes, our \textit{class envy-freeness up to one item} (\CEF{}) requires that envy of any class towards another class to be eliminated after the removal of at most one item that is matched to an agent within the envied class.
%
In the offline setting wherein all items are available up front, it is known that \CEF{} can be achieved without unnecessarily throwing away items~\citep{benabbou2020}.\footnote{We later formalize the latter restriction as \emph{non-wastefulness} (\NW{}). This is required because \CEF{}, on its own, can be achieved vacuously via an empty matching by throwing away all the items.} Can it still be achieved in the online setting?

\paragraph{Impossibility of \CEF{} in online matching.} First, note that a classical algorithm that is blind to the class information can easily violate \CEF{}. For example, if there are two classes containing two agents each, and two items arrive that are liked by all four agents, the algorithm may end up assigning both items to agents from the same class, rendering the other class envious even if we remove one of the items. This simple example is easy to fix via a ``class-aware'' algorithm that pays attention to the classes: simply assign the second item to an agent from the class that did not receive the first item. Alas, a slightly larger example shows that even class-aware online algorithms cannot always achieve \CEF{}. 

\begin{figure}
    \centering
    \footnotesize
    \begin{tikzpicture}[scale=0.6]
        
        \node[draw, circle, fill=pink] (a1) at (-3, 0){$a_{1}$};
        \node[draw, circle, fill=pink] (a2) at (-3, -1.2){$a_{2}$};
        \node[draw, circle, fill=pink] (a3) at (-3, -2.4){$a_{3}$};
        
        \node[draw, circle, fill=cyan] (b1) at (3, 0){$b_{1}$};
        \node[draw, circle, fill=cyan] (b2) at (3, -1.2){$b_{2}$};
        \node[draw, circle, fill=cyan] (b3) at (3, -2.4){$b_{3}$};
        
        \node[draw, circle] (o1) at (0, 0){$o_{1}$};
        \node[draw, circle] (o2) at (0, -1.2){$o_{2}$};
        \node[draw, circle] (o3) at (0, -2.4){$o_{3}$};
        \node[draw, circle] (o4) at (0, -3.6){$o_{4}$};
        
        \path[-, line width=0.5mm] (a1) edge (o1);
        \path[-] (b1) edge (o1);
        
        \path[-] (a2) edge (o2);
        \path[-, line width=0.5mm] (b2) edge (o2);
        
        \path[-] (a3) edge (o3);
        \path[-, line width=0.5mm] (b3) edge (o3);
        
        \path[-] (a1) edge (o4);
        \path[-, line width=0.5mm] (b1) edge (o4);
    \end{tikzpicture}
    \caption{An adversarial instance where \CEF{} cannot be achieved together with non-wastefulness.}
    \label{fig::GEF_non}
\end{figure}
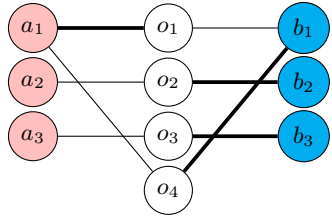

\begin{example} \label{exm:CEF1vNW}
Consider the example shown in \Cref{fig::GEF_non}, in which six agents are partitioned into two classes $N_1 = \{a_1, a_2, a_3\}$ and $N_2 = \{b_1, b_2, b_3\}$, and four items arrive sequentially in the order $(o_1, o_2, o_3, o_4)$. An edge between an agent and an item indicates that the agent likes the item; thick edges indicate the matching. Let us assume that we do not wish to throw away any item as long as there is an unmatched agent who likes it; we later formalize this as \emph{non-wastefulness}. 

For $i \in \set{1,2,3}$, item $o_i$ is liked by agents $a_i$ and $b_i$. The first item $o_1$ can be matched to either $a_1$ and $b_1$; without loss of generality, suppose it is matched to $a_1 \in N_1$. When the second item $o_2$ arrives, note that it must be matched to $b_2 \in N_2$ in order to satisfy \CEF{}. The third item $o_3$ can again be matched to either of $a_3$ and $b_3$; without loss of generality, suppose it is matched to $b_3 \in N_2$. Now, the fourth item $o_4$ arrives, and the algorithm learns that it is liked only by $a_1$ (who is already matched) and $b_1$ (who is unmatched). The algorithm must assign it to $b_1$ due to non-wastefulness, which leaves class $N_1$ envious of class $N_2$, even if we ignore any one of the items assigned to $N_2$.  
\end{example}

Given this impossibility, we seek online matching algorithms that achieve the fairness notions \emph{approximately}, often in conjunction with approximate efficiency guarantees. We aim to answer the following theoretical questions:

\begin{quote} \textit{Can we design deterministic algorithms for matching indivisible or divisible items that achieve approximate class fairness while adhering to efficiency requirements? And, can we surpass their guarantees by using randomization?}
\end{quote}

\subsection{Our Results}
We initiate the study of fairness among classes of agents in online bipartite matching. Our first contribution (\Cref{sec:model}) is developing a detailed mathematical framework in which we adopt classical fairness concepts to online matching. 
We consider two types of online matching models, one with indivisible items, wherein an item must be matched in its entirety to a single agent, and one with divisible items, wherein an item may be fractionally divided between multiple agents.

For both settings, we design online algorithms that achieve approximate fairness and efficiency guarantees, and also provide upper bounds on the approximations that can be achieved by any online algorithm. Our algorithms satisfy non-wastefulness, which implies $\half$-approximation of the optimal utilitarian social welfare (\USW{}); the utilitarian social welfare, i.e., the sum of agent utilities, is effectively the size of the matching. Specifically, we make the following contributions (summarized in \Cref{tab:summary}):

\begin{itemize}
    \item \textbf{Indivisible matching}: When items are indivisible, we develop a deterministic algorithm, \algdetindiv, that simultaneously achieves non-wastefulness, \CEF{\half}, \CMMS{\half}, and \USW{\half} (\cref{thm:CEF1}). The algorithm uses an adaptive priority queue over classes, in which a class is shifted to the end of the queue immediately upon receiving an item. 
    Further, we prove that no deterministic algorithm can achieve any of \CEF{\alpha} (subject to non-wastefulness), \CMMS{\alpha}, or \USW{\half}, for any $\alpha > \half$ (\Cref{thm:CEF1:impossibility}), establishing our algorithm to be simultaneously optimal for each guarantee. 
    
    \item \textbf{Divisible matching}: When items are divisible, we improve the above bounds via a different algorithm, \algdetdiv. This algorithm divides items equally between the classes, but uses water-filling to divide the portion of an item assigned to a class between the agents in that class. This algorithm simultaneously achieves non-wastefulness, \cef{(\e)}, \CPROP{(\e)}, and \USW{\half} (\Cref{thm:divisible}). Furthermore, no deterministic algorithm can achieve \cef{\alpha} for any $\alpha > \sfrac{3}{4}$, or \USW{\alpha} for any $\alpha > \e$, and \CPROP{(\e)} is tight (\cref{thm:divupperbound}).
    
    \item \textbf{Randomized algorithms}: Finally, we propose a randomized algorithm, \algrandindiv, for matching indivisible algorithms that breaks the $\half$ barrier. We run a variant of \algdetdiv to obtain a guiding divisible matching, and round it into an indivisible matching using a technique called online correlated selection (OCS). We prove that it is simultaneously \CPROP{0.593} and \USW{\half} (\Cref{thm:indivisible-randomized}).
\end{itemize}

\begin{table*}[t]
\centering
\resizebox{0.8\textwidth}{!}{%
\footnotesize
\begin{tabular}{llllll}
\multicolumn{3}{c}{\textbf{Indivisible}}                                                                                                                            & \multicolumn{3}{c}{\textbf{Divisible}}                                                             \\ \hline
\multicolumn{1}{|l}{Fairness}                                   & Algorithm     & Upper Bound                                                                        & Fairness          & Algorithm       & \multicolumn{1}{l|}{Upper Bound}                              \\ \hline
\rowcolor[HTML]{CBCEFB} 
\multicolumn{1}{|l}{\cellcolor[HTML]{CBCEFB}$\alpha$-CEF1 + NW} & $\sfrac{1}{2}$ & \multicolumn{1}{l|}{\cellcolor[HTML]{CBCEFB}$\sfrac{1}{2}$}                        & $\alpha$-CEF + NW & $1-\frac{1}{e}$ & \multicolumn{1}{l|}{\cellcolor[HTML]{CBCEFB}$\sfrac{3}{4}$}   \\
\multicolumn{1}{|l}{$\alpha$-CMMS}                              & $\sfrac{1}{2}$ & \multicolumn{1}{l|}{$\sfrac{1}{2}$}                                                & $\alpha$-CPROP    & $1-\frac{1}{e}$ & \multicolumn{1}{l|}{$1-\frac{1}{e}$}                         \\
\rowcolor[HTML]{CBCEFB} 
\multicolumn{1}{|l}{\cellcolor[HTML]{CBCEFB}$\alpha$-USW}       & $\sfrac{1}{2}$ & \multicolumn{1}{l|}{\cellcolor[HTML]{CBCEFB}{\color[HTML]{000000} $\sfrac{1}{2}$}} & $\alpha$-USW      & $\sfrac{1}{2}$   & \multicolumn{1}{l|}{\cellcolor[HTML]{CBCEFB}$1-\frac{1}{e}$} \\ \hline
\end{tabular}%
}
\caption{The summary of our results on deterministic algorithms for matching indivisible and divisible items. Each algorithm achieves its three guarantees simultaneously, while the upper bound holds for any algorithm, separately for each guarantee.}
\label{tab:summary}
\end{table*}

\subsection{Related Work}

\paragraph{Online matching.}
We refer readers to \citet{online_matching_ads} for a survey of the vast literature on online matching, and summarize some results that are the most related to this paper. 
The \ranking algorithm of \citet{karp1990optimal} assigns each item in its entirety; in our model, this corresponds to a randomized algorithm for matching indivisible items that achieves \USW{(\e)}. The case of divisible items is often called \textit{fractional online matching} in the matching literature.\footnote{It is closely related to another model called online $b$-matching in which each offline agent may be matched up to $b$ times. Since the algorithms and analyses are usually interchangeable in these two models, we phrase both models as the case of divisible items.} 
For this, \citet{kalyanasundaram2000optimal} gave a deterministic $(\e)$-competitive algorithm, which achieves \USW{(\e)} in our framework; different papers refer to this algorithm as Balance, Water-filling, or Water-level. 
The \ranking algorithm and its analysis were generalized to the vertex-weighted case by \citet{aggarwal2011online}. 
\citet{feldman2009online} introduced the free disposal model of edge-weighted online matching and gave a $(\e)$-competitive algorithm for divisible items.
The series of works by \citet{fahrbach2020edge}, \citet{shin2021making}, \citet{gao2021improved}, and \citet{blanc2021multiway} led to the state-of-the-art $0.536$-competitive algorithm for edge-weighted online matching with indivisible items.
These works developed a new technique called online correlated selection which we also use in this paper.

The literature also considers stochastic models of online matching problems to break the $\e$ barrier.
\citet{mahdian2011online} and \citet{karande2011online} showed that the competitive ratio of \ranking is between $0.696$ and $0.727$ if online vertices arrive by a random order.
\citet{huang2019online} introduced a variant of \ranking that breaks the $\e$ barrier in vertex-weighted online matching under random-order arrivals; the ratio was further improved to $0.668$~\citep{jin2020improved}.
If items are drawn from a distribution known to the algorithm, it is called \textit{online stochastic matching}~\citep{feldman2009online}.
The best known competitive ratios for unweighted and vertex-weighted online stochastic matching are $0.711$ and $0.700$, respectively~\citep{huang2021online}.

\paragraph{Fair division.}
There is a rich body of literature on fair allocation of indivisible or divisible items. A common assumption in most fair division studies is that there is no constraint on how many items each agent can receive, and agents receive increasing value when receiving more items.

In this literature, envy-freeness and proportionality (and approximations thereof) have been used as the primary criteria of fairness. 
For divisible items, an allocation satisfying both envy-freeness and an economic efficiency notion called Pareto optimality is known to exist~\citep{varian1974equity} and can be computed via convex programming when agents have additive valuations~\citep{eisenberg}.
For indivisible items, two relaxations of envy-freeness are commonly studied: envy-freeness up to one item (EF1)~\citep{lipton2004approximately} and maximin share fairness (MMS)~\citep{budish2011combinatorial}. 
An EF1 allocation is guaranteed to exist with monotone valuations~\citep{lipton2004approximately}, and can be achieved together with Pareto optimality when agents have additive valuations~\citep{caragiannis2016unreasonable}. On the other hand, MMS allocations are not guaranteed to exist, even for additive valuations, though constant factor approximation algorithms~\citep{KurokawaProcacciaJACM,garg2019improved,ghodsi2018fair} and ordinal approximations~\citep{hosseini2021guaranteeing,hosseini2021ordinal} exist and can be computed in polynomial time.

Our problem can be seen as a fair division problem by considering each class to be a meta-agent; the value of this meta-agent for a bundle of items is the maximum total value obtained by matching the items to the agents in the class, which induces OXS valuations~\citep{paesleme2017gs} (these are not additive). \citet{benabbou2019fairness} studied a model similar to ours in the offline setting, and observed that the EF1 algorithm of \citet{lipton2004approximately} may result in a wasteful allocation; nevertheless, they showed that an allocation satisfying EF1 and non-wastefulness exists and can be computed in polynomial time. Subsequent papers~\citep{benabbou2020,babaioff2020fair,BarmanVerma} considered a more general class of submodular valuations with dichotomous marginals and proved that EF1 and optimal \USW{} can be achieved together; \citet{BarmanVerma} proved a similar result for MMS and optimal \USW{}.  

\paragraph{Fairness in online matching.} Our paper is also related to the growing line of work on \emph{online} fair division~\citep{Benade:2018:MEV,gorokh2020online,ZengPsomas2020,walsh2011online,aleksandrov2015online}, but a majority of this work focuses on additive valuations, and hence, their techniques do not apply to our matching setting. Several recent papers are concerned with group fairness in online matching~\citep{MaXu2020,ijcai2021Sankar}. \citet{MaXu2020} studied a stochastic setting wherein the \emph{agents} arrive online (as opposed to the items in our model), following an independent Poisson process with known homogeneous rate; the objective is to maximize the minimum ratio of the number of agents served to the number of agents in each group. \citet{ijcai2021Sankar} studied an online matching problem where the items arrive online. Here, the \emph{items} are 
grouped 
into classes (as opposed to the agents in our model), and each agent specifies capacity constraints, which they referred to as \emph{group fairness constraints}, restricting the number of items from each class that can be assigned to the agent. Due to these crucial differences between their models and ours, their techniques and results do not overlap with ours. 

\section{Model}\label{sec:model}
For $t \in \mathbb{N}$, define $[t] = \set{1,\ldots,t}$. First, let us introduce an offline version of our model and the solution concepts we seek. Later, we will discuss the online model and algorithms in that model. 

Consider a bipartite graph $G = (N, M, E)$, where $N$ represents a set of vertices called agents, $M$ a set of vertices called items, and $E$ the set of edges. We say that agent $a$ {\em likes} item $o$ if $a$ is adjacent to $o$, i.e., $(a,o) \in E$.
The set of agents $N$ is partitioned into $k$ known classes $N_{1}, \ldots, N_{k}$ so that $N_{i}\cap N_{j} = \emptyset$ for all $i \neq j$ and $\cup_{i=1}^k N_i = N$. For simplicity, we refer to class $N_i$ simply as class $i$.

\paragraph{Matching.} We consider the cases of \emph{divisible} items (where each item can be matched to multiple agents fractionally) and \emph{indivisible} items (where each item must be matched to a single agent integrally). A (divisible) \emph{matching} is a matrix $X=(x_{a,o})_{a \in N, o \in M} \in [0,1]^{N \times M}$ satisfying $\sum_{a \in N} x_{a,o} \leq 1$ for each item $o \in M$, and $\sum_{o \in M} x_{a,o} \leq 1$ for each agent $a \in N$. We say that matching $X$ is \emph{indivisible} if $x_{a,o} \in \set{0,1}$ for each agent $a \in N$ and item $o \in M$. Given a matching $X$, we say that agent $a$ is \emph{saturated} if $\sum_{o \in M} x_{a,o} = 1$, and item $o$ is \emph{fully assigned} if $\sum_{a \in N} x_{a,o} = 1$.

For a matching $X$, we write $Y(X)=(\sum_{a \in N_i}x_{a,o})_{i \in [k],o \in M}$ as the matrix containing the total fraction of each item assigned to agents in each class. Let $Y_i(X)$ denote the row of $Y(X)$ corresponding to class $i$. 
For an indivisible matching $X$, we may abuse the notation and use $Y_{i}(X)$ to refer to the set of items matched to agents in class $i$, i.e., $\set{o \in M \mid x_{a,o}=1~\mbox{for some}~a \in N_i}$. We may omit the argument $X$ from $Y(X)$ and $Y_i(X)$ if it is clear from the context. 

\paragraph{Class valuations.} The value derived by agent $a$ from matching $X$ is $V_a(X) = \sum_{o \in M: (a,o) \in E} x_{a,o}$. We define the value of class $i$ from matching $X$ as the utilitarian social welfare of the agents in class $i$ under matching $X$, denoted $V_{i}(X) = \sum_{a \in N_i} V_a(X)$. 

In order to define fairness at the level of classes, we need to also define how much hypothetical value agents in class $i$ could derive from the items matched to agents in another class $j$. However, it is not obvious how one should define this value because it depends on how the items matched to agents in $N_j$ would be matched to agents in $N_i$ in this hypothetical scenario. Following \cite{benabbou2019fairness}, we use the following optimistic valuations.

Given a vector $\bfy=(y_o)_{o \in M} \in [0,1]^M$ representing fractions of different items, the \textit{optimistic valuation} $V^*_i(\bfy)$ of class $i$ for $\bfy$ is the size of the maximum fractional matching between the agents of $N_{i}$ and $\bfy$; namely, $V^*_i(\bfy)$ is given by the optimal value of the following LP: 
\begin{align*}
    \max &~~{\textstyle\sum_{a \in N_i} \sum_{o \in M : (a,o) \in E}}~ x_{a,o}\\
    \mbox{s.t.} &~~\textstyle{\sum_{a \in N_i}}~ x_{a,o} \leq y_o && \forall o \in M,\\
    &~~\textstyle{\sum_{o \in M}}~ x_{a,o} \leq 1 && \forall a \in N_i,\\ 
    &~~x_{a,o} \geq 0 && \forall a \in N_i,o \in M.
\end{align*}

For a set of items $S \subseteq M$, let $\boldsymbol{e}^S \in \{0,1\}^M$ denote the incidence vector such that $\boldsymbol{e}^S_o=1$ if $o \in S$ and $\boldsymbol{e}^S_o=0$ otherwise; we may write $V^*_i(\boldsymbol{e}^S)$ as $V^*_i(S)$ for ease of notation. For an integral vector $\bfy$, it is known that there is an integral optimal solution to the above LP (see, e.g., Section $5$ of~\cite{Korte2006}); thus, $V^*_i(S)$ coincides with the maximum size of an integral matching between $S$ and the agents in $N_i$.  

\subsection{Solution Concepts}
Our goal is to ensure that items are matched to agents in a manner that is fair to agents belonging to different classes. To that end, we consider classical fairness notions from the fair division literature, such as envy-freeness~\citep{george1958puzzle,foley1967resource}, proportionality~\citep{steinhaus1948problem}, and maximin share guarantee~\citep{budish2011combinatorial}, which are typically used to ensure fairness between individual agents. We extend these notions to ensure fairness between classes of agents. 

\paragraph{(Approximate) class envy-freeness.} Envy-freeness between individual agents demands that every agent values the resources allocated to her at least as much as she values the resources allocated to another agent. 
When applied to classes, we compare the value $V_i(X)$ derived by class $i$ for its matched items with class $i$'s optimistic valuation for the items matched to another class $j$, i.e. $V^*_i(Y_j(X))$.
Note that this results in a strong class envy-freeness notion: \emph{even if}, hypothetically, class $i$ were to be matched to the items currently matched to class $j$ under $X$ in an optimal manner, they would still not be any happier overall. 

\begin{definition}[Class envy-freeness]
A matching $X$ is $\alpha$-\emph{class envy-free} (\cef{\alpha}) if for all classes $i,j \in [k]$, $V_{i}(X) \geq \alpha \cdot V_{i}^{*}(Y_{j}(X))$. When $\alpha=1$, we simply refer to it as class envy-freeness (\cef{}). 
\end{definition}

It is impossible to achieve exact \cef{} with an indivisible matching in general.
For example, when one desirable item has to be allocated among two classes, the class which does not receive the item necessarily envies the other class which receives it. Hence, we consider the following relaxation of \cef{} for integral matchings.

\begin{definition}[Class envy-freeness up to one item]
An integral matching $X$ is $\alpha$-\emph{class envy-free up to one item} (\CEF{\alpha}) if for every pair of classes $i,j \in [k]$, either $Y_{j}(X)=\emptyset$ or there exists an item $o \in Y_{j}(X)$ such that $V_{i}(X) \geq \alpha \cdot V_{i}^{*}(Y_{j}(X) \setminus \set{o})$. When $\alpha=1$, we simply refer to it as class envy-freeness up to one item (\CEF{}). 
\end{definition}

We remark that \CEF{} is called type-wise EF1 (TEF1) by \cite{benabbou2019fairness}; we use the terminology ``class'' instead of ``type'' because letting agents of the same ``type'' have different incident edges may be confusing to some readers.

\paragraph{(Approximate) class proportionality and maximin share fairness.} Another classical fairness concept is \emph{proportionality}. In the traditional fair division model where agent valuations are additive and there is no limit to how many items can be assigned to an agent, proportionality is typically stated as requiring that each agent receive value that is at least $\sfrac{1}{n}$-th of her value for the set of all items, where $n$ is the number of agents. This can be equivalently viewed as demanding that each agent receive at least the maximum value she can receive from the worst bundle among all fractional partitions of the items into $n$ bundles. While these two versions are equivalent under additive valuations, they are significantly different under non-additive valuations. For subadditive valuations (like our optimistic valuations), the latter version is stronger. Further, the latter version continues to imply its indivisible counterpart, called maximin share fairness, whereas the former version no longer implies it in our model. For these reasons, we use the latter version as the appropriate definition of proportionality in our model. 

Since we are interested in fairness at the class level, we define the {\em proportional share} of class $i$ as  
\[
\prop_i= \max_{X \in \mathcal{X}} \min_{j \in [k]} V^*_i(Y_j(X)). 
\]
where $\mathcal{X}$ is the set of (divisible) matchings of the set of items $M$ to the set of agents $N$. 

\begin{definition}[Class proportionality]
    We say that matching $X$ is $\alpha$-{\em class proportional} (\CPROP{\alpha}) if for every class $i \in [k]$, $V_i(X) \geq \alpha \cdot \prop_i$. When $\alpha = 1$, we simply refer to it as class proportionality (\CPROP{}). 
\end{definition}

As in the case with class envy-freeness, class proportionality is impossible to guarantee via indivisible matchings. Nevertheless, we can naturally relax the notion of proportionality by only taking into account indivisible matchings in the definition of proportional share above. This naturally adopts the well-studied notion of maximin share fairness to our setting. Formally, the {\em maximin share} of class $i$ is defined as 
\[
\mms_i= \max_{X \in \mathcal{I}} \min_{j \in [k]} V^*_i(Y_j(X)). 
\]
where $\mathcal{I}$ is the set of indivisible matchings of the set of items $M$ to the set of agents $N$.

\begin{definition}[Class maximin share fairness]
    We say that matching $X$ is $\alpha$-{\em class maximin share fair} (\CMMS{\alpha}) if for every class $i \in [k]$, $V_i(X) \geq \alpha \cdot \mms_i$. When $\alpha=1$, we simply refer to it as class maximin share fairness (\CMMS{}).
\end{definition}

For fair division with additive valuations, \citet{segal2019democratic} proved that, subject to allocating every item, EF1 is equivalent to MMS. In contrast, in our model neither implies even an approximation of the other (see \Cref{sec:relation}). 

\paragraph{Efficiency.} We consider two notions of efficiency. \emph{Non-wastefulness} demands that each item to be fully assigned, unless all the agents who like it are saturated. Non-wasteful integral matchings are also known as \emph{maximal} matchings. 

\begin{definition}[Non-wastefulness]
We say that matching $X$ is \emph{non-wasteful} (\NW{}) if there is no pair of agent $a$ and item $o$ such that $a$ likes $o$ (i.e., $(a,o) \in E$), $a$ is not saturated (i.e., $\sum_{o' \in M} x_{a,o'} < 1$), and $o$ is not fully assigned (i.e., $\sum_{a' \in N} x_{a',o} < 1$).
\end{definition}

A more quantitative notion of efficiency is the utilitarian social welfare, which, in our context, is the size of the (divisible) matching. Note that this is the classical objective that the literature on online matching optimizes, in the absence of any fairness constraints. 

\begin{definition}[Utilitarian social welfare]
The {\em utilitarian social welfare} (\USW{}) of a matching $X$ is given by  $\usw(X) = \sum_{a \in N} \sum_{o \in M : (a,o) \in E} x_{a,o}$. We say that a divisible (resp., indivisible) matching $X$ is \USW{\alpha} if $\usw(X) \ge \alpha \cdot \usw(X^*)$ for all divisible (resp., indivisible) matchings $X^*$. When $\alpha=1$, we refer to $X$ as the USW-optimal matching. Note that the benchmarks for the divisible and indivisible cases are identical as the indivisible matching with the highest \USW{} also has the highest \USW{} among all divisible matchings. 
\end{definition}

The following is a known relation between maximal (non-wasteful) and maximum matchings in both divisible and indivisible cases. We provide a proof in the appendix for completeness.
\begin{restatable}{proposition}{NonWastefulHalf}\label{prop:nw-usw}
Every non-wasteful (divisible or indivisible) matching is \USW{\half}.
\end{restatable}

Let us illustrate the above concepts of fairness and efficiency using examples.
\begin{example}
Consider the example given in \Cref{fig:EF_matching}, where there are four items ($o_1$, $o_2$, $o_3$, and $o_4$), agents $a_1$ and $a_2$ belong to one class, and agents $b_1$ and $b_2$ belong to another class. An edge between an agent and an item indicates that the agent likes the item; thick edges indicate matching. \Cref{fig:model-example-a} shows an empty matching, which is class envy-free (\cef{}) but wasteful. 
\Cref{fig:model-example-b} shows a matching that achieves \CEF{} and non-wastefulness. Finally, \Cref{fig:model-example-c} shows a matching that achieves \CEF{} along with optimal utilitarian social welfare. 
\begin{figure}
    \centering
    \footnotesize
    \begin{subfigure}[b]{0.3\linewidth}
    \centering
    \begin{tikzpicture}[scale=0.6]
        
        \node[draw, circle, fill=pink] (a1) at (-3, 0){$a_{1}$};
        \node[draw, circle, fill=pink] (a2) at (-3, -1.2){$a_{2}$};
        
        \node[draw, circle, fill=cyan] (b1) at (3, 0){$b_{1}$};
        \node[draw, circle, fill=cyan] (b2) at (3, -1.2){$b_{2}$};
        
        \node[draw, circle] (o1) at (0, 0){$o_{1}$};
        \node[draw, circle] (o2) at (0, -1.2){$o_{2}$};
        \node[draw, circle] (o3) at (0, -2.4){$o_{3}$};
        \node[draw, circle] (o4) at (0, -3.6){$o_{4}$};
        
        \path[-] (a1) edge (o1);
        \path[-] (a2) edge (o1);
        \path[-] (a2) edge (o2);
        \path[-] (a2) edge (o3);

        \path[-] (b1) edge (o1);
        \path[-] (b1) edge (o3);
        \path[-] (b1) edge (o4);
     
        \path[-] (b2) edge (o2);

    \end{tikzpicture}
    \caption{\cef{} but wasteful}
    \label{fig:model-example-a}
    \end{subfigure}
    \hfill
    \begin{subfigure}[b]{0.3\linewidth}
    \centering
    \begin{tikzpicture}[scale=0.6]
        
        \node[draw, circle, fill=pink] (a1) at (-3, 0){$a_{1}$};
        \node[draw, circle, fill=pink] (a2) at (-3, -1.2){$a_{2}$};

        \node[draw, circle, fill=cyan] (b1) at (3, 0){$b_{1}$};
        \node[draw, circle, fill=cyan] (b2) at (3, -1.2){$b_{2}$};

        \node[draw, circle] (o1) at (0, 0){$o_{1}$};
        \node[draw, circle] (o2) at (0, -1.2){$o_{2}$};
        \node[draw, circle] (o3) at (0, -2.4){$o_{3}$};
        \node[draw, circle] (o4) at (0, -3.6){$o_{4}$};

        \path[-] (a1) edge (o1);
        \path[-] (a2) edge (o1);
        \path[-] (a2) edge (o2);
        \path[-, line width=0.5mm] (a2) edge (o3);
        
        \path[-, line width=0.5mm] (b1) edge (o1);
        \path[-, line width=0.5mm] (b2) edge (o2);
        \path[-] (b1) edge (o3);
        \path[-] (b1) edge (o4);

    \end{tikzpicture}
    \caption{\CEF{} and \NW{}}
    \label{fig:model-example-b}
    \end{subfigure}
    \hfill
     \begin{subfigure}[b]{0.3\linewidth}
     \centering
    \begin{tikzpicture}[scale=0.6]
        
        \node[draw, circle, fill=pink] (a1) at (-3, 0){$a_{1}$};
        \node[draw, circle, fill=pink] (a2) at (-3, -1.2){$a_{2}$};

        \node[draw, circle, fill=cyan] (b1) at (3, 0){$b_{1}$};
        \node[draw, circle, fill=cyan] (b2) at (3, -1.2){$b_{2}$};

        \node[draw, circle] (o1) at (0, 0){$o_{1}$};
        \node[draw, circle] (o2) at (0, -1.2){$o_{2}$};
        \node[draw, circle] (o3) at (0, -2.4){$o_{3}$};
        \node[draw, circle] (o4) at (0, -3.6){$o_{4}$};

        \path[-, line width=0.5mm] (a1) edge (o1);
        \path[-] (a2) edge (o1);
        \path[-] (a2) edge (o2);
        \path[-, line width=0.5mm] (a2) edge (o3);
        
        \path[-] (b1) edge (o1);
        \path[-, line width=0.5mm] (b2) edge (o2);
        \path[-] (b1) edge (o3);
        \path[-, line width=0.5mm] (b1) edge (o4);

    \end{tikzpicture}
    \caption{\CEF{} and \USW{1}}
    \label{fig:model-example-c}
    \end{subfigure}
    \caption{Class envy-freeness (\cef{}), non-wastefulness (\NW{}), and utilitarian social welfare approximation (\USW{}).}
    \label{fig:EF_matching}
\end{figure}
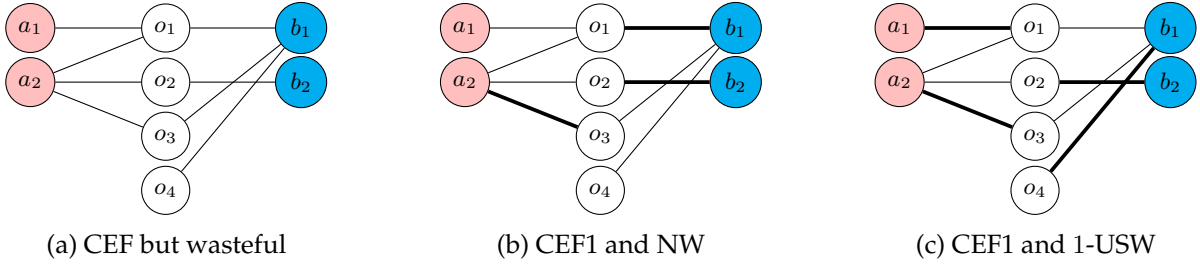

\end{example}

\subsection{Online Model} 
Let us now introduce our online model. In this model, the items in $M$ arrive one-by-one in an arbitrary order. We refer to the step in which item $o \in M$ arrives as step $o$. 

When item $o$ arrives, all agents reveal whether or not they like the item. In other words, the edges incident to item $o$ are revealed in graph $G$. At this point, an online algorithm must make an immediate and irrevocable decision to ``match'' the item to the agents in $N$, i.e., set the values of $(x_{a,o})_{a \in N}$. We consider both algorithms which set these values \emph{deterministically} and ones which set them in a \emph{randomized} fashion (but must fix them before the next item arrives). For randomized algorithms, we seek the desired guarantees in expectation.

For the algorithms we design in this paper, we prove that they achieve the desired guarantees (approximate \cef{}, \CEF{}, \CPROP{}, \CMMS{}, \USW{}, or non-wastefulness) at every step. However, a key property of our algorithms is that they do not need to know in advance the number of items that will arrive, which means that proving the desired guarantees at the end implies that that they hold at every step. In contrast, our upper bounds (impossibility results) will hold even if the desired guarantees are required to hold only at the end.

\begin{definition}
For $\alpha \in (0,1]$, a deterministic online algorithm for matching divisible or indivisible items is \cef{\alpha} (resp., \CEF{\alpha}, \CPROP{\alpha}, \CMMS{\alpha}, \USW{\alpha}, or \NW{}) if it produces an \cef{\alpha} (resp., \CEF{\alpha}, \CPROP{\alpha}, \CMMS{\alpha}, \USW{\alpha}, or \NW{}) matching when all items have arrived.
\end{definition}

\begin{definition}
For $\alpha \in (0,1]$, a randomized online algorithm for matching indivisible items is 
\begin{itemize}
\item  \cef{\alpha} if, when all items have arrived, it produces a matching $X$ such that for every pair of classes $i,j\in [k]$, $\mathbb{E}[V_{i}(X)] \geq \alpha \cdot \mathbb{E}[V_{i}^{*}(Y_{j}(X))]$;
\item  \CPROP{\alpha} if, when all items have arrived, it produces a matching $X$ such that for every class $i \in [k]$, $\mathbb{E}[V_{i}(X)] \geq \alpha \cdot \prop_i$; and
\item \USW{\alpha} if, when all items have arrived, it produces a matching $X$ such that $\E[\usw(X)] \ge \alpha \cdot \usw(X^*)$, where $\E[\usw(X)] = \sum_{a \in N} \sum_{o \in M : (a,o) \in E} \E[x_{a,o}]$ and $X^*$ is a matching with the highest utilitarian social welfare. 
\end{itemize}
\end{definition}

Because \CMMS{} and \CPROP{} place only a lower bound on the utility of every agent, there is no tension between them and non-wastefulness. Any algorithm achieving an approximation of these notions can be made non-wasteful without losing the said fairness approximation. We provide a formal proof in the appendix.

\begin{restatable}{proposition}{CMMSNonWastefull}\label{prop:MMS:NW}
For $\alpha \in (0,1]$, if there is a deterministic online algorithm satisfying \CMMS{\alpha} (resp., \CPROP{\alpha}), then there is a non-wasteful deterministic online algorithm satisfying \CMMS{\alpha} (resp., \CPROP{\alpha}). This holds for matching both divisible and indivisible items. 
\end{restatable}

\section{Deterministic Algorithms for Indivisible Items}
\label{sec:indivisible}

We start by focusing on deterministic algorithms for matching indivisible items. We study possible approximations of two fairness concepts, \CEF{} and \CMMS{}, along with efficiency guarantees in terms of non-wastefulness and the utilitarian social welfare.

When matching indivisible items, \CEF{} may seem trivial to achieve: only match an item to some agent in some class if this preserves \CEF{}, and discard the item otherwise. However, this algorithm may `waste' too many items and lose significant efficiency.\footnote{In fact, discarding all items---an empty matching---is vacuously class envy-free.}

\cref{exm:CEF1vNW} illustrated that \CEF{} and non-wastefulness are incompatible in the online setting.\footnote{In \cref{sec:pessimistic}, we show that this incompatibility holds even after weakening the \CEF{} requirement to account for `pessimistic' valuations, i.e, when each class evaluates the items matched to another class through a \textit{minimum-cardinality maximal matching}.}
In this light, for arbitrary classes, it is natural to ask what approximation of \CEF{} can be achieved subject to non-wastefulness. 

\subsection{Algorithm \algdetindiv}

One way to achieve approximate \CEF{} is to ensure a balanced treatment of all classes by providing them approximately equal `opportunity' for receiving an item. This approach is inspired by the well-studied \textit{Round-Robin} algorithm in fair division~\citep{caragiannis2016unreasonable} and its widely-adopted cousin, \textit{Draft}, that is used in sports for selecting players~\citep{brams1979prisoners,brams2000win} or assigning courses to college students~\citep{budish2012multi}.

However, running such algorithms na\"{\i}vely in our online setting, where not all items are available upfront, can be problematic: if we do a round-robin over classes, a class can be disadvantaged if the item arriving in its turn is not liked by any unmatched agent in the class. Further, non-wastefulness requires that any arriving item be matched as long as there is an unsaturated agent who likes it, even if this agent does not belong to the class whose turn it is. Keeping these observations in mind, we design \algdetindiv (\Cref{alg:indivisible}), which provides equal treatment to the different classes while achieving non-wastefulness.

\begin{algorithm}[t]
\caption{\algdetindiv}
\label{alg:indivisible}
\DontPrintSemicolon
Fix a priority ordering over classes, $\pi = (\pi_1,\ldots, \pi_{k})$ \;
\When{item $o\in M$ arrives}{
    \For{$i = 1$ to $k$}{
        Let $N_{\pi_{i}, o}$ be the set of unmatched agents $a \in N_{\pi_{i}}$ such that $(a,o) \in E$ \;
        \If{$N_{\pi_{i}, o} \neq \emptyset$}{
            Arbitrarily match $o$ to an agent in $N_{\pi_{i}, o}$ \;
            $\pi \leftarrow (\pi_{1},\ldots,\pi_{i-1}, \pi_{i+1}, \ldots, \pi_{k}, \pi_{i})$ \;
            \Break\;
            }
        }
}
\end{algorithm}

\paragraph{Algorithm description.} Fix an arbitrary priority ordering $\pi=(\pi_{1}, \pi_{2}, \ldots, \pi_{k})$ over the $k$ classes, where $\pi_{1}$ is the class with the highest priority.
Upon arrival of each item, pick the first class $N_{\pi_{i}}$ in the priority ordering that contains an unmatched agent who likes the item. Match the item to any unmatched agent---there may be several such agents---in $N_{\pi_{i}}$ who likes the item. Update the priority ordering $\pi$ by moving class $\pi_i$ to the end.

The following theorem establishes approximate fairness and efficiency guarantees of \algdetindiv; later, in \Cref{thm:CEF1:impossibility}, we prove that these guarantees are tight.

\begin{theorem}\label{thm:CEF1}
For deterministic matching of indivisible items, \algdetindiv (\Cref{alg:indivisible}) satisfies non-wastefulness, \CEF{\half}, \CMMS{\half}, and \USW{\half}.
\end{theorem}
\begin{proof} Let $X$ be the matching returned by the algorithm at the end.

\paragraph{\NW{}} Non-wastefulness of $X$ follows immediately from the description of the algorithm: at each step, the arriving item is matched to an agent who likes it whenever such an agent exists. 

\paragraph{\USW{}} Because $X$ is non-wasteful, due to \Cref{prop:nw-usw} it also satisfies \USW{\half}. 

\medskip\noindent Now, we turn our attention to the fairness guarantees. Recall that for each $i \in [k]$, $Y_i$ denotes the set of items matched to agents in class $j$. Fix any class $i$. Let $t = |Y_i|$ denote the number of items matched to the agents in class $i$ under $X$. Due to non-wastefulness, we have $V_i(X)=t$. 

\paragraph{\CEF{\half}.} Consider any class $j \in [k] \setminus \{i\}$. Let $Y^*_j \subseteq Y_j$ be the set of items matched to class $j$ that are liked by at least one unmatched agent in class $i$. The claim immediately holds when $Y^*_j =\emptyset$: in this case, the optimistic value of class $i$ for $Y_j$ is $V^*_i(Y_j) \le t = V_i(X)$, implying that $X$ satisfies \cef{} for $i$. Thus, we assume that at least one item in $Y_j$ is liked by at least one unmatched agent of class $i$. 

By construction of the algorithm, we have $|Y^*_j| \leq t+1$. This is because every time class $j$ receives an item in $Y^*_j$ (that is liked by an agent in class $i$ who remains unmatched till the end, and, therefore, is unmatched at the time of the item's arrival), class $j$ must have a higher priority than class $i$. Hence, the algorithm must match an item to class $i$ before it can match another item in $Y^*_j$ to class $j$. Thus, $|Y^*_j| \le 1+|Y_i| = t+1$.

Fix an arbitrary item $o \in Y^*_j \subseteq Y_j$. We claim that $V^*_i(Y_j \setminus \set{o}) \le 2t$, which establishes the \CEF{\half} claim. Note that the $t$ matched agents in class $i$ can derive a maximum total utility of $t$ from these items. Further, the total utility that the unmatched agents in class $i$ can derive from these items is upper bounded by $|Y^*_j \setminus \set{o}| \le t$. Hence, $V^*_i(Y_j \setminus \set{o}) \le 2t$.

\paragraph{\CMMS{\half}.} Assume for contradiction that $t=V_i(X) < (\half) \cdot \mms_i$. Because $\mms_i$ is an integer, this implies $2t + 1 \leq \mms_i$. Let $(S_1,S_2,\ldots,S_k)$ be a maximin partition of the items for class $i$ such that $V^*_i(S_j) \geq \mms_i$ for every $j \in [k]$. By our assumption, we have $V^*_i(S_j) \geq 2t+1$ for every $j \in [k]$. For each $j \in [k]$, we let $S^*_j$ denote the set of items in $S_j$ that are liked by at least one unmatched agent in class $i$. Note that $V^*_i(S_j) \le t+|S^*_j|$: the $t$ matched agents in class $i$ can derive total utility at most $t$, and the unmatched agents can derive total utility at most $|S^*_j|$.

Recall that $|Y_i| = t$ and we have already established $|Y^*_j| \leq t+1$ for every class $j \in [k] \setminus \{i\}$. Further, by non-wastefulness, none of the unmatched agents of class $i$ likes any item in $O \setminus \bigcup_{h \in [k]} Y_h$. Thus, we have $|\bigcup_{j \in [k]} S^*_j| \leq |Y_i \cup (\bigcup_{j \in [k]\setminus \{i\}} Y^*_j)| \leq t+ (k-1)(t+1)$, meaning that there exists some $h \in [k]$ such that $|S^*_h| \leq t$. Thus, we have $V^*_i(S_h) \leq 2t < 2t+1$,
a contradiction. 
\end{proof}

Before we turn to proving these guarantees to be the best possible in our online setting, we remark that in the offline setting, it is known that (exact) \CEF{} and \NW{} can be achieved simultaneously~\cite{benabbou2019fairness}. However, whether they can be achieved together with \CMMS{\alpha}, for any $\alpha > 0$, is an interesting open question. 

\subsection{Impossibility Results}

In this section, we show that the each of the fairness and efficiency guarantees achieved by \algdetindiv (\Cref{thm:CEF1}) is tight; no deterministic online algorithm for matching indivisible items can achieve a better approximation. Note that our \CEF{} upper bound is subject to non-wastefulness because an algorithm can trivially achieve \CEF{} on its own by throwing away every item. 

The constructions are based on creating instances in which a subset of agents in one class get saturated early on, rendering the class envious of another class at the end since all the remaining items can only be matched to the agents in that other class.

\begin{theorem}\label{thm:CEF1:impossibility}
No deterministic online algorithm for matching indivisible items can achieve any of the following guarantees:
\begin{itemize}
    \item \CEF{\alpha} for any $\alpha > \half$ and non-wastefulness,
    \item \CMMS{\alpha} for any $\alpha > \half$, 
    \item \USW{\alpha} for any $\alpha > \half$.
\end{itemize}
\end{theorem}

\begin{proof}
We argue each impossibility result separately.

\paragraph{\CEF{} and \NW{}} Consider \Cref{exm:CEF1vNW} in the introduction. In that example, we already argued that any deterministic online algorithm satisfying non-wastefulness ends up matching (without loss of generality) $Y_2 = \set{o_2,o_3,o_4}$ to class $2$ and $Y_1 = \set{o_1}$ to class $1$. One can check that $V^*_1(Y_2 \setminus \set{o}) = 2$ for any $o \in Y_2$, whereas $V_1(X) = 1$, implying that the algorithm cannot achieve \CEF{\alpha} for any $\alpha > \half$.

\paragraph{\CMMS{}} We will prove that no deterministic online algorithm \emph{satisfying non-wastefulness} can achieve \CMMS{\alpha} for any $\alpha > \half$. \Cref{prop:MMS:NW} implies that no deterministic algorithm, regardless of whether it satisfies non-wastefulness, can guarantee \CMMS{\alpha} for any $\alpha > \half$. 

Since we have assumed non-wastefulness, we can repeat the construction used above for the \CEF{} upper bound. Consider the same example again, and consider the partition the items into $(\tilde{Y}_1 = \set{o_1,o_2}, \tilde{Y}_2 = \set{o_3,o_4})$. Note that $V^*_1(\tilde{Y}_1) = V^*_1(\tilde{Y}_2) = 2$, implying that the maximin share of class $1$ is $\mms_1 \ge 2$. Since the value derived by class $1$ is $V_1(X) = 1$, we see that the algorithm cannot achieve \CMMS{\alpha} for any $\alpha > \half$.

\paragraph{\USW{}} Note that the \USW{} guarantee does not depend on the class structure; hence, the well-known upper bound of $\half$ on the approximation of a maximum matching by any deterministic algorithm carries over to our model, and implies the desired \USW{\half} upper bound. For completeness, consider the following simple instance. 

There are two items, $o_1$ and $o_2$, arriving in the increasing order of their indices. There is a single class containing two agents. Item $o_1$ is liked by both agents. The algorithm matches it to one of the two agents. Item $o_2$ then arrives, and is liked only by the agent who did not receive item $o_1$. The optimal utilitarian social welfare is $2$, but that of the algorithm is only $1$. 
\end{proof}

Following \Cref{thm:CEF1:impossibility}, a natural question is whether there is any way to circumvent this impossibility result. We show that two such approaches do not work, demonstrating robustness of \Cref{thm:CEF1:impossibility}.

\begin{remark}[Reshuffling items within each class cannot help.]
One idea is to only require the online algorithm to match each item to a class, and allow every class to optimally distribute the items matched to it among its members at the end. This effectively increases the utility of class $i$ from $V_i(X)$ to $V^*_i(Y_i)$. However, in \Cref{exm:CEF1vNW} used for the \CEF{} and \CMMS{} upper bounds in the proof above, the matching produced already assigns items optimally within each class (i.e., satisfies $V_i(X) = V^*_i(Y_i)$ for each class $i$). Hence, reshuffling items at the end cannot improve the value any further. 
This shows that we must use randomization when deciding which class should receive an item in order to achieve a better approximation; this is precisely what we achieve in \Cref{sec:randomized}.
\end{remark}

\begin{remark}\label{remark:allocated:MMS}
Another natural direction is to weaken the requirements in \Cref{thm:CEF1:impossibility}. In our online setting, there is a weakening of our \CMMS{\alpha} guarantee that also makes sense. Instead of computing the MMS values by partitioning the set of \emph{all} items, we can first observe the matching $X$ produced by an algorithm and then compute the MMS values by having each class partition only the set of items \emph{allocated} under $X$. This produces smaller (or equal) values, making this \CMMS{} with respect to \emph{allocated items} a weaker requirement than our \CMMS{} with respect to \emph{all items}. 

\algdetindiv achieves a $\half$-approximation of the stronger requirement. In contrast, the proof of \Cref{thm:CEF1:impossibility} shows that no non-wasteful\footnote{Seeking the weaker requirement makes sense only with non-wastefulness since the empty matching vacuously satisfies it.} algorithm can achieve $(\half+\epsilon)$-approximation of even the weaker requirement, for any $\epsilon > 0$, because all items are allocated in our construction. 
\end{remark}

\section{Deterministic Algorithms for Divisible Items}
\label{sec:divisible}

We now turn our attention to deterministic online matching of \emph{divisible} items. First, we design an algorithm that simultaneously achieves non-wastefulness, \cef{(\e)}, \CPROP{(\e)}, and \USW{\half}. Later, we prove upper bounds on the approximation ratio of each guarantee that hold for any algorithm. 

\subsection{Algorithm \algdetdiv}

We propose an algorithm, \algdetdiv (presented as \Cref{alg:divisible}), that divides items equally at the class level and performs water-filling to further divide the items assigned to each class between the agents in that class. Recall that our model has a capacity constraint: $\sum_{o \in M} x_{a,o} \le 1$ for each agent $a$. Agent $a$ is saturated if $\sum_{o \in M} x_{a,o} = 1$, and unsaturated otherwise. 

When item $o$ arrives, \algdetdiv continuously splits the item equally among classes with at least one unsaturated agent who likes the item.\footnote{We do not yet need to know how the fraction of item $o$ assigned to a class is divided between its members; we can simply keep track of the total remaining capacity of the agents in the class who like the item.} At the end of this process, each class either receives the same fraction $\beta_o$ of the item, or has all of its agents who like item $o$ saturated. This computation is performed in \Cref{alg:divisible:classphase} of \Cref{alg:divisible}. Then, to divide fraction of item $o$ assigned to each class $i$ within its members, we conduct water-filling among the members who like item $o$, which continuously prioritizes agents with the lowest utility. At the end of this process, each member who likes item $o$ either receives the same final utility $\gamma_{i,o}$ or is saturated. This computation is performed in \Cref{alg:divisible:individual-phase} of \Cref{alg:divisible}. 

\begin{algorithm}[htb]
\caption{\algdetdiv}
\label{alg:divisible}
\DontPrintSemicolon
Initialize $X = ( x_{a,o} )_{a \in N, o \in M}$ so that $x_{a,o}=0$ for every agent $a$ and item $o$\;
Initialize $Y = (y_{i,o})_{i \in [k], o \in M}$ so that $y_{i,o}=0$ for every class $i$ and item $o$\;
\When{item $o\in M$ arrives}{
\textbf{/*class-phase*/}\;
Define the demand of each class $i \in [k]$ as $d_{i,o} = \sum_{a \in N_{i,o}} (1 - \sum_{o' \in M} x_{a,o'})$\;
Find the largest $\beta_o \le 1$ satisfying $\sum_{i \in [k]} \min \{\beta_o, d_{i,o} \} \le 1$\;
Set $y_{i,o} = \min \{ \beta_o, d_{i,o}\}$ for each $i \in [k]$\;\label{alg:divisible:classphase}
\For{$i = 1$ to $k$}{
\textbf{/*individual-phase*/}\;
    Let $N_{i,o}$ denote the set of neighbours of item $o$ in class $i$, i.e., $N_{i,o}=\{a \in N_i: (a,o) \in E \}$ \;
    Find the largest $\gamma_{i,o} \le 1$ satisfying $\sum_{j \in N_{i,o}} \max\set{\gamma_{i,o} - \sum_{o' \in M} x_{a,o'},0} \le y_{i,o}$\;
    Set $x_{a,o} = \max\set{\gamma_{i,o} - \sum_{o' \in M} x_{a,o'},0}$ for all $a \in N_{i,o}$\;\label{alg:divisible:individual-phase}
    }
}
\end{algorithm}

\begin{theorem}\label{thm:divisible}
For deterministic matching of divisible items, \algdetdiv (\Cref{alg:divisible}) satisfies non-wastefulness, \cef{(\e)}, \CPROP{(\e)}, and \USW{\half}.
\end{theorem}
\begin{proof}
    We prove that \algdetdiv satisfies each of the desirable properties.
    
    \paragraph{\NW{}} Non-wastefulness follows by the algorithm's definition.

    \paragraph{\USW{\half}} This is implied by non-wastefulness (\Cref{prop:nw-usw}).

    \paragraph{\cef{(\e)}}
    Consider two arbitrary classes $i$ and $j$. We want to prove that class $i$'s value for its matching is at least $1-\frac{1}{e}$ times its optimistic value for class $j$'s matching, i.e., $V_i(X) \ge (\e) \cdot V_i^*(Y_j)$. 
    %

    For $\theta \in [0,1]$, let $f(\theta)$ denote the number of agents in class $i$ who have value (``water level'') at least $\theta$ under $X$. Let $N_i(\theta)$ be the set of these $f(\theta)$ agents and $\bar{N}_i(\theta) = N_i \setminus N_i(\theta)$. One can check that for any $\theta \in [0,1]$, $\int_0^{\theta} f(z) \dif z = \sum_{a \in N_i} \min(\theta,\sum_{o \in M} x_{a,o})$. 

    Let us now rewrite both $V_i(X)$ and $V^*_i(Y_j)$ in terms of $f(y)$. Plugging in $\theta = 1$ above, we see that the total value of the agents in class $i$ is given by 
    \[
        V_i(X) = \int_0^1 f(z) \dif z.
    \]

    Next, fix an arbitrary $\theta \in (0,1]$. In order to upper bound $V^*_i(Y_j)$, we consider the value derived from $Y_j$ by the agents in $N_i(\theta)$ and those in $\bar{N}_i(\theta)$. 
    
    Since agents in $\bar{N}_i(\theta)$ remain unsaturated till the end, for every item $o$ liked by any such agent, the fraction $y_{i,o}$ of the item given to class $i$ must be at least as much as the fraction $y_{j,o}$ of it given to class $j$. Further, the portion given to class $i$ must be assigned to agents who, at the time of the assignment, had value less than $\theta$. Hence, the total fraction of items given to class $j$ that are liked by at least one agent in $\bar{N}_i(\theta)$, which is an upper bound on the contribution of the agents in $\bar{N}_i(\theta)$ to $V^*_i(Y_j)$, is at most $\int_0^\theta f(z) \dif z$. Note that the $f(\theta)$ agents in $N_i(\theta)$ contribute at most $1$ each to $V^*_i(Y_j)$. Combining these observations, the optimistic value of class $i$ for the items assigned to class $j$ satisfies
    \[
        V_i^*(Y_j) \le \int_0^\theta f(z) \dif z + f(\theta),
        \qquad
        \forall 0 < \theta \le 1.
    \]

    Multiplying the above inequality by $e^{\theta-1}$ and integrating over $\theta \in (0, 1]$, we get:
    \begin{align*}
        \left(1-\frac{1}{e}\right) V_i^*(Y_j) &= \int_{\theta=0}^1 e^{\theta-1}~V^*_i(Y_j) \dif \theta \\
        &\le \int_{\theta=0}^1 e^{\theta-1} \left(\int_{z=0}^\theta f(z) \dif z + f(\theta)\right) \dif \theta\\
        &= \int_{z=0}^1 f(z) \left(\int_{\theta=z}^1 e^{\theta-1} \dif \theta\right) \dif z + \int_{\theta=0}^1 e^{\theta-1} f(\theta) \dif \theta\\
        &= \int_{z=0}^1 \left(1-e^{z-1}\right) f(z) \dif z +  \int_{z=0}^1 e^{z-1} f(z) \dif z\\
        &= \int_{z=0}^1 f(z) \dif z = V_i(X),
    \end{align*}
    where the third transition follows from breaking the integral over the two terms and exchanging the order of integrals in the first part; and during the fourth transition, we rename the index from $\theta$ to $z$ in the second part.

    \paragraph{\CPROP{(\e)}}
    Consider an arbitrary class $i$. We want to prove that class $i$'s value for the matching is at least $\e$ times its proportional share, i.e., $V_i(X) \ge (\e) \cdot \prop_i$. Consider an arbitrary divisible partition of the items $\tilde{Y}$, consisting of non-negative vectors $\tilde{Y}_i = (\tilde{y}_{i,o})_{o \in M}$ for $i \in [k]$ satisfying $\sum_{i \in [k]} \tilde{y}_{i,o} = 1$ for each $o \in M$. It suffices to prove that:
    \[
        k \cdot V_i(X) \ge \left(1-\frac{1}{e}\right)\cdot \sum_{j \in [k]} V_i^*(\tilde{Y}_j).
    \]

    Recall that $f(\theta)$ denotes the number of agents in class $i$ who have value at least $\theta$ under $X$, $N_i(\theta)$ is the set of these $f(\theta)$ agents, and $\bar{N}_i(\theta) = N_i \setminus N_i(\theta)$. Fix an arbitrary $\theta \in (0,1]$. 
    
    Since the agents in $\bar{N}_i(\theta)$ remain unsaturated till the end, for each item $o$ liked by at least one such agent, the algorithm gives $y_{i,o} \ge 1/k$ fraction of the item to class $i$ (but not necessarily to the agents in $\bar{N}_i(\theta)$). Further, as argued above, this portion of the item must be assigned to the agents in the class who, at the time of the assignment, have value less than $\theta$. Hence, the total number of items liked by at least one agent in $\bar{N}_i(\theta)$, which is an upper bound on the contribution of these agents to $\sum_{j \in [k]} V^*_i(\tilde{Y}_j)$, is at most $k \int_0^\theta f(z) \dif z$. 

    Also, each of $f(\theta)$ many agents in $N_i(\theta)$ can contribute a value of at most $1$ to $V^*_i(\tilde{Y}_j)$ for each $j \in [k]$. Hence, the total contribution of these agents to $\sum_{j \in [k]} V^*_i(\tilde{Y}_j)$ is at most $k \cdot f(\theta)$. 

    Combining the two observations, we get that 
    \[
        \sum_{j \in [k]} V_i^*(\tilde{Y}_j) \le k \cdot \left( \int_0^\theta f(z) dz + f(\theta) \right), \quad \forall 0 < \theta \le 1.
    \]

    Multiplying the inequality by $e^{\theta-1}$, integrating over $\theta \in [0, 1]$, and following the same steps as in the \cef{(\e)} proof above, we have:
    \[
        \left( 1 - \frac{1}{e} \right) \cdot \sum_{j \in [k]} V_i^*(\tilde{Y}_j) \le k \cdot \int_0^1 f(z) \dif z = k \cdot V_i(X),
    \]
    as needed.
\end{proof}

\subsection{Impossibility Results}
Our goal in this section is to provide upper bounds on the fairness and efficiency guarantees that hold for any deterministic online algorithm for matching divisible items. We prove that the \CPROP{(\e)} guarantee achieved by \algdetdiv is tight, and establish a weaker upper bound on \cef{} and \USW{}.

\begin{theorem}\label{thm:divupperbound}
No deterministic online algorithm for matching divisible items can achieve any of the following guarantees:
\begin{itemize}
    \item \cef{\alpha} for any $\alpha > \sfrac{3}{4}$ and non-wastefulness,
    \item \CPROP{\alpha} for any $\alpha > \e$, 
    \item \USW{\alpha} for any $\alpha > \e$.
\end{itemize}
\end{theorem}
\begin{proof}
We argue each impossibility separately. 

\paragraph{\cef{} and \NW{}} Consider any deterministic online algorithm that satisfies non-wastefulness. Consider an instance that consists of two classes, $N_1=\{a_1,a_2,a_3\}$ and $N_2=\{b_1,b_2,b_3\}$, and four items $o_1,o_2,o_3,o_4$ arriving in that order. We denote by $X$ the matching that will be produced by the algorithm on this instance. 

Agents $a_1$, $a_2$, $b_1$, and $b_2$ like the first two items $o_1$ and $o_2$. By non-wastefulness, the algorithm must fully divide $o_1$ and $o_2$ between $\{a_1,a_2,b_1,b_2\}$. Without loss of generality, suppose that the total fraction of these items assigned to class $N_1$ is at least the total fraction assigned to class $N_2$, i.e., $\sum_{a \in N_1} \sum_{o \in \set{o_1,o_2}} x_{a,o}  \geq \sum_{b \in N_2} \sum_{o \in \set{o_1,o_2}} x_{b,o}$. Further, we assume, without loss of generality, that agent $b_1$ obtains at least as much total fraction of these items as agent $b_2$, i.e., $\sum_{o \in \set{o_1,o_2}} x_{b_1,o} \geq \sum_{o \in \set{o_1,o_2}} x_{b_2,o}$. Finally, all agents of class $N_1$ as well as agent $b_1$ like the remaining two items $o_3$ and $o_4$; agents $b_2$ and $b_3$ do not like them. We will prove that $V_2(X) \leq (\sfrac{3}{4}) \cdot V^*_2(Y_1)$. 

First, we show that $V_2(X) \leq \sfrac{3}{2}$. Observe that the value derived by $b_2$ under $X$ is at most $\sfrac{1}{2}$. This holds because the total fraction of $o_1$ and $o_2$ assigned to $b_2$ is at most $\sfrac{1}{2}$ by the assumptions above, and the agent does not like items $o_3$ and $o_4$. Further, agent $b_3$ does not like any of the items. Thus, the total value class $N_2$ can achieve under $X$ is $V_2(X) \le 1+\sfrac{1}{2} = \sfrac{3}{2}$. 

Next, we show that $V^*_2(Y_1) \geq 2$. Note that $N_1$ must receive a total fraction of at least $1$ from each of $\{o_1,o_2\}$ and $\{o_3,o_4\}$. Since $b_2$ likes every item in $\{o_1,o_2\}$ and $b_1$ likes every item in $\{o_3,o_4\}$, class $N_2$ can optimistically derive a total value of at least $2$ by assigning $Y_{1,o_1}$ and $Y_{1,o_2}$ fractions of $o_1$ and $o_2$ to $b_2$ (capped by $1$), and $Y_{1,o_3}$ and $Y_{1,o_4}$ fractions of $o_3$ and $o_4$ to $b_1$ (capped by $1$). 

This shows that the algorithm does not achieve \cef{\alpha} for any $\alpha > \sfrac{3}{4}$. 

\paragraph{\USW{}}
Note that the utilitarian social welfare is simply the size of the (divisible) matching, which is independent of the class information. Hence, the $\e$ upper bound on \USW{} follows from the classical $\e$ upper bound on the competitive ratio of any online divisible matching algorithm; see, e.g., the work of \citet{kalyanasundaram2000optimal}.

\paragraph{\CPROP{}} Consider an instance of a single class. In this case, the proportional share of the class coincides with the value $\usw(X^*)$ of a \USW{}-optimal matching $X^*$. Thus, the $\e$ upper bound on \CPROP{} approximation follows from the $\e$ upper bound on \USW{} approximation.
\end{proof}

\begin{remark}
Similar to Remark~\ref{remark:allocated:MMS}, one may wonder what we can say about a weaker notion of proportionality with respect to only the \emph{allocated items}, i.e., if the proportional share of each class is defined based on the divisible matchings of the allocated items (instead of all items). In \Cref{prop:impossibility:prop:allocated} in Appendix~\ref{app:divisible}, we show that the upper bound of $\e$ continues to hold even for this weaker version. However, unlike in the case of indivisible items, this does not immediately follow from the proof above (which considers an instance with a single class, for which, trivially, the weaker version is exactly satisfied). The proof of \Cref{prop:impossibility:prop:allocated} is much more intricate. 
\end{remark}

While \algdetdiv achieves the optimal $\e$ approximation of \CPROP{}, its guarantees with respect to \cef{} and \USW{} identified in \Cref{thm:divisible} are weaker than the upper bounds in \Cref{thm:divupperbound}. One might wonder if this is simply because our analysis in \Cref{thm:divisible} is loose. We show that this is not the case. Hence, future work must focus either on proving better upper bounds, or on designing new algorithms which might surpass \algdetdiv. 

\begin{restatable}{proposition}{divisibleUpperBound}\label{prop:alg2:EF}
\algdetdiv does not achieve any of the following guarantees:
\begin{itemize}
    \item \cef{\alpha} for any $\alpha > \e$, 
    \item \CPROP{\alpha} for any $\alpha > \e$,
    \item \USW{\alpha} for any $\alpha > \half$.
\end{itemize}
\end{restatable}

\section{Randomized Algorithms for Indivisible Items} 
\label{sec:randomized}

Recall from \Cref{sec:indivisible} that for indivisible items, no deterministic online algorithm can achieve \CMMS{\alpha} for any $\alpha > \half$. When moving to randomized algorithms, one can naturally hope to approximate \CPROP{} instead of \CMMS{} because the value to a class is evaluated in expectation. However, apriori it is not clear whether a randomized algorithm can achieve \CPROP{\alpha} for any $\alpha > 1/2$.

By applying a recently introduced rounding technique, called \textit{Online Correlated Selection} (OCS) \citep{fahrbach2020edge}, to the divisible matching given by \algdetdiv (Algorithm~\ref{alg:divisible}), we are able to design a randomized algorithm for indivisible items that achieves \CPROP{0.593}.

We start by introducing a recent result about OCS that forms the backbone of our approach.

\begin{lemma}[c.f., \citealt{gao2021improved}]
    \label{lem:semi-ocs}
    There is a polynomial-time online algorithm which works as follows. In each step, it takes as input a non-negative vector $(\tilde{x}_{a,o})_{a \in N}$ for some $o \in M$ satisfying $\sum_{a \in N} \tilde{x}_{a,o} \le 1$ and selects an agent $a$ with positive $\tilde{x}_{a,o}$. Further, by the end, each agent $a$ is selected at least once with probability at least:
    \[
        p(\tilde{x}_a) = 1 - \exp\left( - \tilde{x}_a - \tfrac{1}{2} \cdot \tilde{x}_a^2 - \tfrac{4-2\sqrt{3}}{3} \cdot \tilde{x}_a^3 \right),
    \]
    where $\tilde{x}_a = \sum_{o \in M} \tilde{x}_{a,o}$.
\end{lemma}

Technically, such an algorithm is called (multi-way) semi-OCS instead of OCS. But the nomenclature is unimportant for our application, so we will call it OCS for brevity, and refer interested readers to the works of \citet{fahrbach2020edge} and \citet{gao2021improved} for a detailed comparison.

\emph{How good is the guarantee in \Cref{lem:semi-ocs}?} For comparison, consider the simpler independent randomized rounding algorithm, which, upon receiving the vector $(\tilde{x}_{a,o})_{a \in N}$, selects each agent $a$ with probability $\tilde{x}_{a,o}$, independently of the rounding outcomes in the previous steps. By the end, each agent $a$ is selected at least once with probability $1-\prod_{o \in M} (1-\tilde{x}_{a,o}) \ge 1-\exp(-\sum_{o \in M} \tilde{x}_{a,o}) = 1-\exp(-\tilde{x}_a)$. Readers can verify that using this weaker bound in the proof of \Cref{thm:indivisible-randomized} only yields  \CPROP{\half}. The improved guarantee in \Cref{lem:semi-ocs} is critical for achieving an approximation better than $\half$.

%
Our algorithm, \algrandindiv (presented as \Cref{alg:indivisible-randomized}), runs a variant of \algdetdiv in the background to get a guiding divisible matching $\tilde{X} = (\tilde{x}_{a,o})_{a \in N, o \in M}$. The only difference is that unlike \algdetdiv, this variant does not cap the value (total fraction of all items) assigned to an agent at $1$. This is because the algorithm will perform rounding to compute an indivisible matching, and by \Cref{lem:semi-ocs}, the probability that an agent $a$ is matched depends on the value $\tilde{x}_a$ of the agent in the divisible matching in such a manner that even reaching a value of $1$ would not guarantee being matched with certainty. 

Upon receiving a new item $o$, the algorithm first continues running this variant of \algdetdiv to obtain the guiding division $(\tilde{x}_{a,o})_{a \in N}$ (Lines~\ref{alg:indivisible-randomized:class}-\ref{alg:indivisible-randomized:beforeOCS}), and then lets OCS select an agent $a^*$ accordingly (\Cref{alg:indivisible-randomized:OCS}). If the selected agent $a^*$ is not yet matched, the algorithm matches item $o$ to this agent. If $a^*$ is already matched, the algorithm matches item $o$ to an arbitrary unmatched agent who likes it, and discards the item if there is no such agent (\Cref{alg:indivisible-randomized:match}). 

\begin{algorithm}
\DontPrintSemicolon
\caption{\algrandindiv}
\label{alg:indivisible-randomized}
    Initialize an empty indivisible matching $X = (x_{a,o})_{a \in N, o \in M}$\;
    Initialize an empty divisible matching $\tilde{X} = (\tilde{x}_{a,o})_{a \in N, o \in M}$\;
    Maintain a class-level divisible matching $\tilde{Y} = (\tilde{y}_{i,o} = 0)_{i \in [k], o \in M}$ such that $y_{i,o} = \sum_{a \in N_i} \tilde{x}_{a,o}$\;
    \When{item $o\in M$ arrives}{
        \smallskip
        \textbf{/*class-phase divisible matching*/}\;\label{alg:indivisible-randomized:class}
        For each class $i$, let $N_{i,o}$ be the set of agents in class $i$ who like item $o$\;
        Let $k_o$ be the number of classes $i$ such that $N_{i,o} \neq \emptyset$\;
        Let $\tilde{y}_{i,o} = \frac{1}{k_o}$ for each of these $k_o$ classes\;
        \smallskip
        \textbf{/*individual-phase divisible matching*/}\;
        \For{each class $i$ with $y_{i,o}>0$}{
            Find $\gamma_o$ such that $\sum_{a \in N_{i,o}} \max(\gamma_o - \tilde{x}_a,0) = \tilde{y}_{i,o}$\;
            Let $\tilde{x}_{a,o} = \max(\gamma_o - \tilde{x}_a,0)$ for all $a \in N_{i,o}$\;\label{alg:indivisible-randomized:beforeOCS}
        }
        \smallskip
        \textbf{/*indivisible matching rounded by OCS*/}\;
        %
        Send $(\tilde{x}_{a,o})_{a \in N}$ to the OCS in Lemma~\ref{lem:semi-ocs} and let it select an agent $a^*$\;\label{alg:indivisible-randomized:OCS}
        Match $o$ to $a^*$ if $a^*$ is not yet matched, and to an arbitrary unmatched neighbor (if any) otherwise\;\label{alg:indivisible-randomized:match}
    }
    \smallskip
\end{algorithm}

\begin{theorem}\label{thm:indivisible-randomized}
For randomized matching of indivisible items, \algrandindiv (\Cref{alg:indivisible-randomized}) satisfies non-wastefulness, \CPROP{0.593}, and \USW{\half}.
\end{theorem}

\begin{proof}
    Non-wastefulness is clear from \Cref{alg:indivisible-randomized:match} of \Cref{alg:indivisible-randomized}. \Cref{prop:nw-usw} implies \USW{\half}. Hence, we focus on the interesting \CPROP{0.593} guarantee. 
    
    Fix an arbitrary class $i$. The first part of the analysis bounds the proportional value of class $i$ using the guiding divisible matching $\tilde{X}$. This part is almost verbatim to its counterpart in the proof of \Cref{thm:divisible}, except we do not bound the value threshold $\theta$ by $1$. We include this part to be self-contained.

    For $\theta \ge 0$, let $f(\theta)$ denote the number of agents in class $i$ who have value at least $\theta$ under $\tilde{X}$. Let $N_i(\theta)$ denote the set of these $f(\theta)$ agents, and let $\bar{N}_i(\theta) = N_i \setminus N_i(\theta)$.
    
    Fix any $\theta > 0$. For each item $o$ liked by at least one agent in $\bar{N}_i(\theta)$, \Cref{alg:indivisible-randomized} assigns a fraction $\tilde{y}_{i,o} \ge \sfrac{1}{k_{o}}$ to class $i$ in the guiding divisible matching (but not necessarily to the agents in $\bar{N}_i(\theta)$). Further, any agent in $N_i$ receiving a positive share of item $o$ must have value less than $\theta$ right after receiving it. Hence, the total number of items liked by at least one agent in $\bar{N}_i(\theta)$ is at most $k \int_0^\theta f(z) \dif z$. 
    
    On the other hand, the total value that agents in $N_i(\theta)$ can obtain from any set of items is at most $f(\theta)$ (at most $1$ per agent).
    
    Therefore, for any divisible partition of the items, denoted by non-negative vectors $\hat{Y}_i = (\hat{y}_{i,o})_{o \in M}$ for $i \in [k]$ such that $\sum_{i \in [k]} \hat{Y}_{i,o} = 1$ for each $o \in M$, we have:
    \[
        \sum_{j \in [k]} V_i^*(\hat{Y}_j) \le k \cdot \left( \int_0^\theta f(z) \dif z + f(\theta) \right), \quad \forall \theta > 0.
    \]

    This implies that the proportional share of $i$ is bounded by:
    \begin{equation}
        \label{eqn:indivisible-mms-bound}
        \prop_i \le \int_0^\theta f(z) \dif z + f(\theta), \quad \forall \theta > 0.
    \end{equation}

    Next, we lower bound the expected value of class $i$ for the randomized indivisible matching $X$. OCS ensures that for each agent $a$ in class $i$, its probability of being matched is at least $p(\tilde{x}_a)$.  Hence, the expected value of class $i$ for $X$ is:
    \begin{align*}
        \E[V_i(X)]
        &
        \ge \sum_{a \in N_i} p(\tilde{x}_a)
        \tag{Lemma~\ref{lem:semi-ocs}} \\
        &
        = - \int_0^\infty p(\theta) \dif f(\theta)
        \tag{definition of $f(\theta)$} \\
        &
        = \int_0^\infty p'(\theta) f(\theta) \dif \theta
        ~.
        \tag{integration by parts, $p(0) = f(\infty) = 0$}
    \end{align*}

    Multiplying inequality~\eqref{eqn:indivisible-mms-bound} by non-negative coefficients $c(\theta)$ (to be determined later), and integrating over $\theta > 0$ gives that:
    \begin{align*}
        \prop_i \cdot \int_0^\infty c(\theta) \dif\theta
        &
        \le \int_0^\infty c(\theta) \left( \int_0^\theta f(z) \dif z + f(\theta) \right) \dif \theta \\
        &
        = \int_0^\infty c(\theta) \int_0^\theta f(z) \dif z\, \dif \theta + \int_0^\infty c(\theta) f(\theta) \dif \theta
        \\
        &
        = \int_0^\infty \left( \int_z^\infty c(\theta) \dif \theta + c(z) \right) f(z) \dif z,
    \end{align*}
    where, during the last transition, we exchange the order of integrals in the first part and change the index from $\theta$ to $z$ in the second part.

    We choose $c(\theta) = - e^\theta \int_\theta^\infty p''(y) e^{-y} \dif y$, so that $\int_z^\infty c(\theta) \dif \theta + c(z) = p'(z)$ for all $z > 0$. Hence, we get that:
    \[
         \prop_i \cdot \int_0^\infty c(\theta) \dif \theta
 \le \int_0^\infty p'(z) f(z) \dif z \le \E\,V_i(X) 
        ~.
    \]
    
    The theorem then follows by numerically calculating the integral:
    \[
        \int_0^\infty c(\theta) \dif \theta \approx 0.5936 > 0.593.
    \]
    
    This concludes the proof of the theorem.
\end{proof}

In \Cref{app:randomized}, we briefly discuss other randomized algorithms and their obstacles in achieving better than $\half$ approximation to \CPROP{}. We also present a randomized algorithm based on the classical \ranking algorithm, which achieves \cef{(\e)}. While it achieves this guarantee non-vacuously (i.e., it does not simply return the empty matching), it still violates non-wastefulness. It would be interesting to analyze its efficiency. 

\section{Discussion}

Our work introduces the novel framework of class fairness in online matching. We derive bounds on approximate fairness and efficiency guarantees that deterministic and randomized online algorithms can achieve in this framework for matching divisible and indivisible items, and leave open a number of exciting open questions. For example, can a deterministic algorithm for matching divisible items achieve a \cef{} approximation together with non-wastefulness better than $\e$? (We conjecture the answer to be \emph{no}.) Can it achieve any reasonable \cef{} or \CPROP{} approximation together with a \USW{} approximation better than $\half$ (ideally, $\e$)? Can a randomized algorithm for matching indivisible items achieve any reasonable \cef{} approximation together with either non-wastefulness or a \USW{} approximation? 

More broadly, our basic framework paves the road for interesting extensions. For example, one can allow agents to have non-binary values for the items, consider class fairness notions that give more importance to bigger classes, consider both agents and items arriving online~\citep{huang2020fully}, study weaker adversarial models, or consider stochastic instead of adversarial arrivals. 

All of these fall under the umbrella of online fair allocation of \emph{private goods}, which is a literature still in its infancy with many exciting research directions in sight. Studying its counterpart, online fair allocation of \emph{public goods}, is another worthy goal, which may bring its own set of challenges.

\section*{Acknowledgments}
Hadi Hosseini acknowledges support from NSF IIS grants \#2052488 and \#2107173. Zhiyi Huang was supported in part by an RGC grant \#17201221. Ayumi Igarashi was supported by JST PRESTO under grant number JPMJPR20C1. Nisarg Shah was partially supported by an NSERC Discovery Grant.

\bibliographystyle{plainnat}
\bibliography{abb,references,online}

\clearpage
\appendix
\section*{Appendix}

\section{Omitted Material from Section~\ref{sec:model}}

\NonWastefulHalf*

\begin{proof}[Proof of \Cref{prop:nw-usw}]
Let $X^*$ be a matching maximizing the utilitarian social welfare. Without loss of generality, we can pick $X^*$ to be integral. Let $X$ be any non-wasteful (divisible or indivisible) matching. Hence, for every $(a,o) \in E$, we have $\sum_{o' \in M} x_{a,o'} = 1$ or $\sum_{a' \in N} x_{a',o}=1$. Then, we have
\begin{align*}
    \usw(X^*) &= \sum_{(a,o) : x^*_{a,o} = 1} 1 \le \sum_{(a,o) : x^*_{a,o} = 1} \left( \sum_{o' \in M} x_{a,o'} + \sum_{a' \in N} x_{a',o} \right)\\
    &\le \sum_{a \in N} \sum_{o' \in M} x_{a,o'} + \sum_{o \in M} \sum_{a' \in N} x_{a',o} = 2 \cdot \usw(X),
\end{align*}
where the second transition holds because $x^*_{a,o}=1$ implies $(a,o) \in E$ and $X$ is non-wasteful, and the third transition holds because $X^*$ is an indivisible matching (i.e., if $x^*_{a,o}=1$, $x^*_{a',o'} = 1$, and $(a,o) \neq (a',o')$, then $a \neq a'$ and $o \neq o'$). This proves that $X$ is \USW{\half}.
\end{proof}

\CMMSNonWastefull*
\begin{proof}[Proof of \Cref{prop:MMS:NW}]
    Let us first consider indivisible items. Let $A$ be any deterministic online algorithm that may be wasteful.
    Consider a non-wasteful version of it, denoted as $A'$, that works as follows.
    It runs $A$ in the background and treats $A$'s output as an advice.
    Importantly, $A$ keeps its own internal state and is oblivious to the actual matching decisions made by $A'$.
    For an item $o$, suppose that $A$ matches $o$ to agent $a$.
    Algorithm $A'$ would follow $A$'s advice and match $o$ to $a$ if $a$ is not yet matched, and would otherwise match $o$ to any unmatched agent who likes item $o$.
    
    By definition, $A'$ is non-wasteful.
    Further, we can prove by induction over the steps that the set of agents matched by $A'$ is 
    a superset of the set of agents matched by $A$. Since \CMMS{} is a monotone property (i.e., increasing agent values preserves its approximation), $A'$ achieves at least as good an approximation of \CMMS{} as $A$ does.
    
    For divisible items, the same proof works for \CPROP{}, except $A'$ now gives a fraction of $o$ to each agent $a$ that is the minimum of the fraction of $o$ matched to $a$ under the advice given by $A$ and the remaining capacity of $a$ in the current matching maintained by $A'$. 
\end{proof}

\section{Omitted Material from Section~\ref{sec:indivisible}}

\subsection{Pessimal class envy-freeness} \label{sec:pessimistic}

One may wonder whether relaxing the way each class measures its hypothetical value for a set of items could help alleviating the incompatibility between class envy-freeness and non-wastefulness. We show that even if each class considers a pessimistic value for a set of items (in other words, considers worst-case scenario for matching the items), the clash between envy-freeness and non-wastefulness persists.

Given a vector $\bfy=(y_o)_{o \in M} \in \{0,1\}^M$ representing a set of items, the \textit{pessimistic valuation} $V^{\ominus}_i(\bfy)$ of class $i$ for $\bfy$ is the value of a \textit{minimum-cardinality maximal matching} between the agents of $N_{i}$ and the set $\{\, o \in M \mid y_o=1 \,\}$. 
This problem has shown to be NP-hard for graphs with maximum degree 3 and $k$-regular bipartite graphs for $k\geq 3$ \citep{demange2008minimum,yannakakis1980edge}.

We compare the value $V_i(X)$ derived by class $i$ from matching $X$ with class $i$'s pessimistic valuation for the items matched to another class $j$, i.e. $V^{\ominus}_i(Y_j(X))$.

\begin{definition}[Pessimal class envy-freeness]
A matching $X$ is $\alpha$-\emph{pessimal class envy-free} (\pef{\alpha}) if for every pair of classes $i,j \in [k]$, $V_{i}(X) \geq \alpha \cdot V_{i}^{\ominus}(Y_{j}(X))$. When $\alpha=1$, we simply refer to it as pessimal class envy-freeness (\pef{}). 
\end{definition}

Similar to its optimistic counterpart, \cef{}, a \pef{} matching may not always exist. Therefore, we consider the following relaxation of \pef{} for integral matchings.

\begin{definition}[Pessimal class envy-freeness up to one item]
An integral matching $X$ is $\alpha$-\emph{pessimal class envy-free up to one item} (\PEF{\alpha}) if for every pair of classes $i,j \in [k]$, either $Y_{j}(X)=\emptyset$ or there exists an item $o \in Y_{j}(X)$ such that $V_{i}(X) \geq \alpha \cdot V_{i}^{\ominus}(Y_{j}(X) \setminus \set{o})$. When $\alpha=1$, we simply refer to it as class envy-freeness up to one item (\PEF{}). 
\end{definition}

It is easy to verify that \PEF{} is weaker than \CEF{}. Intuitively, a class values its matching compared to the items assigned to another class if it has a pessimistic view of the items arrival and matched items, should the items were exchanged.
Clearly, a \cef{} matching is also \pef{}, and similarly \CEF{} implies \PEF{}.

\begin{example}
In the example given in Figure~\ref{fig:OEFvPEF}, there are two classes $N_1=\{a_1,a_2\}$ and $N_2=\{b_1,b_2\}$. The bold edges indicate the matched items. This matching is not \cef{}, since class $N_1$ envies class $N_2$ should it able to optimally match items $o_1$ and $o_3$ within its members.
However, the same matching is \pef{} because class $N_1$ considers a pessimal matching of the same items, that is $o_1$ and $o_3$, where item $o_1$ is matched to $a_{1}$ upon its arrival, and thus, $o_{3}$ remains unmatched (Since there is no edge from $a_2$ to $o_3$).
 \begin{figure}
    \centering
    \footnotesize
    \begin{tikzpicture}[scale=0.6]
        
        \node[draw, circle, fill=pink] (a1) at (-3, 0){$a_{1}$};
        \node[draw, circle, fill=pink] (a2) at (-3, -1.2){$a_{2}$};

        \node[draw, circle, fill=cyan] (b1) at (3, 0){$b_{1}$};
        \node[draw, circle, fill=cyan] (b2) at (3, -1.2){$b_{2}$};

        \node[draw, circle] (o1) at (0, 0){$o_{1}$};
        \node[draw, circle] (o2) at (0, -1.2){$o_{2}$};
        \node[draw, circle] (o3) at (0, -2.4){$o_{3}$};

        \path[-] (a1) edge (o1);
        \path[-, line width=0.5mm] (a1) edge (o2);
        \path[-] (a1) edge (o3);
        \path[-] (a2) edge (o1);
     
        \path[-, line width=0.5mm] (b1) edge (o1);
        \path[-, line width=0.5mm] (b2) edge (o3);
        \path[-] (b1) edge (o2);

    \end{tikzpicture}
    \caption{An allocation that is \pef{} but not \cef{}. The red group pessimally considers the worst-case matching of items $o_1$ and $o_3$ with the value of 1.}
    \label{fig:OEFvPEF}
\end{figure}
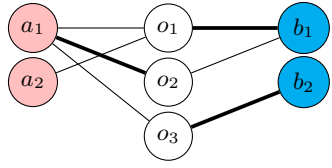

\end{example}

The following proposition strengthens our previous results on the incompatibility between non-wastefulness and \CEF{} by showing that non-wastefulness remains incompatible with a weaker fairness notion of \PEF{}.

\begin{proposition} \label{prop:CEF1_NW}
No deterministic algorithm for matching indivisible items can guarantee non-wastefulness and \PEF{}. 
\end{proposition}

\begin{proof}
Consider the example given in \cref{fig::GEF_non}. It is easy to verify that the matching is non-wasteful. However, in this scenario the pessimal value of class $N_1$ for the items assigned to the class $N_2$ is $3$, implying that the matching is not \PEF{}.
\end{proof}

\subsection{Relationships Between \CEF{} and \CMMS{}}\label{sec:relation}

\begin{proposition}[\CEF{}+\NW{} $\nRightarrow$ \CMMS{}]
\label{prop:ef1NotCMMS}
Given an indivisible instance, a \CEF{}+\NW{} matching does not imply any \CMMS{\alpha} for any $\alpha > 0$.
\end{proposition}
\begin{proof}
We construct an instance for which a \CEF{\alpha}+\NW{} matching with $\alpha = 1$ gives only a \CMMS{0} approximation.

Suppose there are $k$ classes $N_1,N_2,\ldots,N_k$. Each $N_i$ for $i \in [k-1]$ consists of $k$ agents. The last class $N_k$ consists of $k-1$ agents $a_1,a_2,\ldots,a_{k-1}$. 
There are $k(k-1)$ items that are partitioned into $k-1$ subsets $C_1,C_2,\ldots,C_{k-1}$. For $j \in [k-1]$, $C_j$ consists of $k$ items, $o_{1j},o_{2j},\ldots,o_{kj}$, each of which is referred to as a {\em type $j$} item. 
For each $j \in [k-1]$, every agent in class $N_j$ likes every item in $C_j$. 
For class $N_k$, each agent $a_j$ for $j\in [k-1]$ likes every item in $C_j$. For example, agent $a_2$ likes $k$ items $o_{12},o_{22},\ldots,o_{k2}$ but does not like none of the other items. 

\begin{figure*}
    \centering
    \begin{tikzpicture}[scale=0.6, every node/.style={scale=0.6}]
    \tikzstyle{A} = [draw=black,ultra thick,fill=white,isosceles triangle,isosceles triangle apex angle=60,rotate=90,minimum size=15pt]
    \tikzstyle{B} = [draw=black,ultra thick,fill=white,circle,minimum size=15pt]
    \tikzstyle{C} = [draw=black,ultra thick,fill=white,minimum size=15pt]
    \tikzstyle{brace} = [decorate,decoration={calligraphic brace,amplitude=8pt},ultra thick]
    \draw[dashed,thick] (-.8,-.6) rectangle +(4.6,2) node[above=5pt,midway] {\LARGE class $k$};
    \node[A] (A0) at (0,0) {};
    \node[B] (B0) at (1,.1) {};
    \node at (2,.1) {\textbf{\ldots}};
    \node[C] (C0) at (3,.1) {};
    \draw[dashed,thick] (4.2,-.6) rectangle +(4.6,2) node[above=5pt,midway] {\LARGE class $1$};
    \node[A] (A1) at (5,0) {};
    \node[A] (A2) at (6,0) {};
    \node at (7,.1) {\textbf{\ldots}};
    \node[A] (Ak) at (8,0) {};
    \draw[brace] (4.2,1.8) --  +(4.6,0) node[above=10pt,midway] {\LARGE $k$ agents};
    \draw[dashed,thick] (9.2,-.6) rectangle +(4.6,2) node[above=5pt,midway] {\LARGE class $2$};
    \node[B] (B1) at (10,.1) {};
    \node[B] (B2) at (11,.1) {};
    \node at (12,.1) {\textbf{\ldots}};
    \node[B] (Bk) at (13,.1) {};
    \draw[brace] (9.2,1.8) --  +(4.6,0) node[above=10pt,midway] {\LARGE $k$ agents};
    \node at (15,.1) {\LARGE \textbf{\ldots}};
    \draw[dashed,thick] (16.2,-.6) rectangle +(4.6,2) node[above=5pt,midway] {\LARGE class $k-1$};
    \node[C] (C1) at (17,.1) {};
    \node[C] (C2) at (18,.1) {};
    \node at (19,.1) {\textbf{\ldots}};
    \node[C] (Ck) at (20,.1) {};
    \draw[brace] (16.2,1.8) --  +(4.6,0) node[above=10pt,midway] {\LARGE $k$ agents};
    %
    \draw[dashed,thick] (4.2,-6.2) rectangle +(4.6,2) node[below=4pt,midway] {\LARGE type $1$};
    \node[A] (a1) at (5,-5) {};
    \node[A] (a2) at (6,-5) {};
    \node at (7,-4.9) {\textbf{\ldots}};
    \node[A] (ak) at (8,-5) {};
    \draw[brace] (8.8,-6.6) --  +(-4.6,0) node[below=10pt,midway] {\LARGE $k$ items};
    \draw[dashed,thick] (9.2,-6.2) rectangle +(4.6,2) node[below=4pt,midway] {\LARGE type $2$};
    \node[B] (b1) at (10,-4.9) {};
    \node[B] (b2) at (11,-4.9) {};
    \node at (12,-4.9) {\textbf{\ldots}};
    \node[B] (bk) at (13,-4.9) {};
    \draw[brace] (13.8,-6.6) --  +(-4.6,0) node[below=10pt,midway] {\LARGE $k$ items};
    \node at (15,-4.9) {\LARGE \textbf{\ldots}};
    \draw[dashed,thick] (16.2,-6.2) rectangle +(4.6,2) node[below=4pt,midway] {\LARGE type $k-1$};
    \node[C] (c1) at (17,-4.9) {};
    \node[C] (c2) at (18,-4.9) {};
    \node at (19,-4.9) {\textbf{\ldots}};
    \node[C] (ck) at (20,-4.9) {};
    \draw[brace] (20.8,-6.6) --  +(-4.6,0) node[below=10pt,midway] {\LARGE $k$ items};
    \draw[thick] (A1) -- (a1);
    \draw[thick] (A2) -- (a2);
    \draw[thick] (Ak) -- (ak);
    \draw[thick] (B1) -- (b1);
    \draw[thick] (B2) -- (b2);
    \draw[thick] (Bk) -- (bk);
    \draw[thick] (C1) -- (c1);
    \draw[thick] (C2) -- (c2);
    \draw[thick] (Ck) -- (ck);
\end{tikzpicture}
    \caption{A \CEF{}+\NW{} matching that does not imply any approximation for \CMMS{}.}
    \label{fig:CEF1vMMS}
\end{figure*}
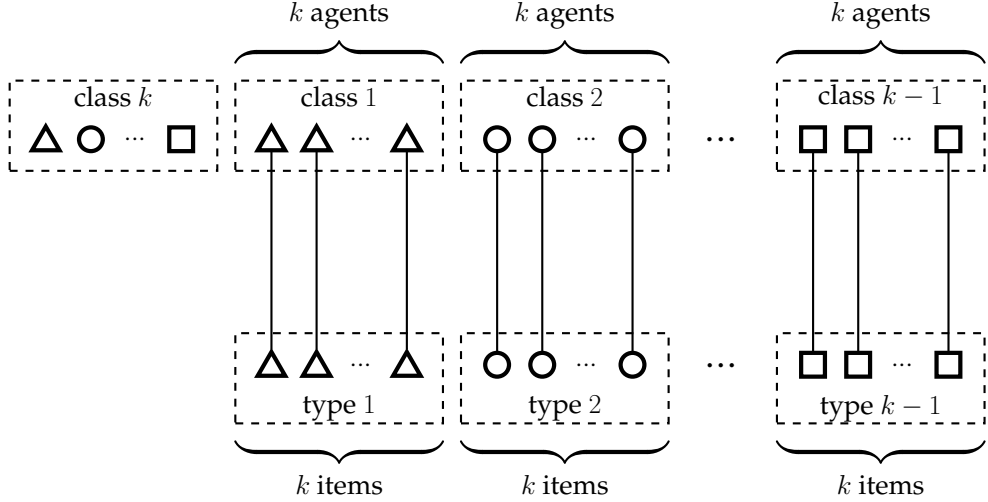

Now, consider a matching $X$ that gives no item to class $N_k$ and matches arbitrarily each of the $k$ items in $C_j$ to one of the $k$ agents in each class $N_j$ for $j \in [k-1]$ (as illustrated in \cref{fig:CEF1vMMS}). Since each of the $k(k-1)$ items are fully assigned to an agent who likes it, the matching $X$ is clearly non-wasteful. Further, this matching is \CEF{}. In fact, all classes except $N_k$ receive a perfect matching and are not envious of any other class. Also, for $j \in [k-1]$, there is at most one agent $a_j$ in $N_k$ who likes an item in $C_j$. Thus, class $N_k$ is not envious for more than one item since $V^{*}_{k}(Y_j) \leq 1$ for any $j\in [k-1]$. Thus, the matching is \CEF{}.

In contrast, consider a partition $(L_1,L_2,\ldots,L_k)$ of the items where $L_i=\{o_{i1},o_{i2},\ldots,o_{ik-1}\}$ for each $i \in [k]$. Observe that for each $i=1,2,\ldots,k$, each agent $a_j$ in $N_k$ likes exactly one item $o_{ij}$ in $L_i$, i.e., $L_i \cap C_j = \{o_{ij}\}$ for $j \in [k-1]$. This means that there is a perfect matching of size $k-1$ between $N_k$ and the items of each $L_i$, yielding $V^*_1(L_i) \geq k-1$ for $i \in [k]$. We thus establish that $\mms_k \geq k-1$. Given that class $N_1$'s value for $X$ is $V_1(X) = 0$, $X$ provides \CMMS{0} approximation, which proves the claim.  
\end{proof}

\begin{proposition}[\CMMS{} $\nRightarrow$ \CEF{}+\NW{}] \label{prop:cmmsNotef1}
Given an indivisible instance, a \CMMS{} matching does not imply \CEF{\alpha} for any $\alpha > 0$.
\end{proposition}

\begin{proof}
Consider an instance with $k$ classes each with $k$ agents. There are $k-1$ items liked by every agent in each class. A matching that assigns all $k-1$ of items to a single class, say $N_1$, satisfies \CMMS{}. This is because the \CMMS{} value for each class is obtained by partitioning the $k-1$ items into $k$ bundles, yielding $\mms_i=0$ for $i=1,2,\ldots,k$.
However, this matching is not \CEF{} (nor any $\alpha$ approximation of it for $\alpha > 0$) because every class values the matching assigned to $N_1$ as $k-1$ while only receiving $0$ valuation.
\end{proof}

\section{Omitted Material from Section~\ref{sec:divisible}}\label{app:divisible}
\subsection{Proportionality with respect to allocated items}
Our objective of this section is to show that $(\e)$-bound is tight even for \CPROP{} with respect to the allocated items. Formally, we define the {\em proportional share} of class $i$ with respect to a set $S$ of items as  
\[
\prop^S_i= \max_{X \in \mathcal{X}(S)} \min_{j \in [k]} V^*_i(Y_j(X)). 
\]
where $\mathcal{X}(S)$ is the set of (divisible) matchings of the set of items $S$ to the set of agents $N$. For $\alpha \in (0,1]$, we say that matching $X$ is $\alpha$-{\em class proportional} (\CPROP{\alpha}) with respect to  a set $S$ of items if for every class $i \in [k]$, $V_i(X) \geq \alpha \cdot \prop^S_i$. For $\alpha \in (0,1]$, a deterministic online algorithm for matching divisible items is $\alpha$-{\em class proportional} (\CPROP{\alpha}) {\em with respect to the allocated items} if when all items have arrived, it produces a matching that is $\alpha$-{class proportional} with respect to the items that have been fully assigned by the algorithm. 

\begin{proposition}\label{prop:impossibility:prop:allocated}
No deterministic algorithm for matching divisible items satisfies \CPROP{\alpha} with respect to the allocated items for any $\alpha > \e$. 
\end{proposition}
\begin{proof}
We will prove that no deterministic online algorithm satisfying non-wastefulness can achieve \CPROP{\alpha} with respect to the allocated items for any $\alpha > \e$. By the proof of~\Cref{prop:MMS:NW}, this implies that no deterministic algorithm can guarantee \CPROP{\alpha} with respect to the allocated items for any $\alpha > \e$. 

Take any non-wasteful algorithm for divisible item allocation and consider the following adversarial instance. There are two classes of $3n$ agents each, $N_1 = \{a_1,\ldots,a_n,d_1,\ldots,d_{2n}\}$ and $N_2 = \{a'_1,\ldots,a'_n,d'_1,\ldots,d'_{2n}\}$. We call the agents $d_1,\ldots,d_{2n},d'_1,\ldots,d'_{2n}$ \emph{dummy agents}. There are $2n$ items, labeled $o_i$ and $o'_i$ for $i \in [n]$.

The construction of the instance works in rounds as follows.
\begin{itemize}
    \item We start with $t=1$, $R^0_1 = \{a_1,a_2,\ldots,a_n\}$, and $R^0_2 = \{a'_1,a'_2,\ldots,a'_n\}$.
    \item In round $t$, items $o_t$ arrives, followed immediately by item $o'_t$. Both these items are liked by agents in $R^{t-1}_1$ and $R^{t-1}_2$. 
    \item Let $V^t(a)$ denote the value that agent $a$ derives at the end of round $t$ when the algorithm finishes allocating both items. Find the lowest valuation agent in each class. WLOG, say $a_t \in \argmin_{a \in R^{t-1}_1} V^t(a)$ and $a'_t \in \argmin_{a' \in R^{t-1}_2} V^t(a')$. Set $R^t_1 \gets R^{t-1}_1 \setminus \{a_t\}$, $R^t_2 \gets R^{t-1}_2 \setminus \{a'_t\}$, and $t \gets t+1$.
\end{itemize}

We stop this process after the first round $t^*$ such that at the end of that round every agent in $R^{t^*}_1$ and every agent in $R^{t^*}_2$ is fully saturated. 

Without loss of generality, assume that at the end of round $t^*$, the total value of agents in $N_1$ is at most the total value of agents in $N_2$, i.e., $\sum_{a \in N_1} V^{t^*}(a) \le \sum_{a' \in N_2} V^{t^*}(a')$. For shorthand, let us denote $V^t(A) = \sum_{a \in A} V^t(a)$ for a set of agents $A$. 

Then, the remaining $2(n-t^*)$ items that arrive are liked by agents in $N_2 \cup R^{t^*}_1$. Note that by non-wastefulness and by the fact that $N_2$ contains $2n$ dummy agents, the $2(n-t^*)$ items are fully assigned to some agent. 

We claim the following properties at the end of round $t^*$.
\begin{itemize}
    \item The agents $n-t^*$ agents in $R^{t^*}_1$ and the $n-t^*$ agents in $R^{t^*}_2$ are all fully saturated.
    \item $V^{t^*}(N_1) \le t^*$, $V^{t^*}(N_2) \ge t^*-1$.
    \item $t^* \le (1-1/e) \cdot n$ (in particular, the process will stop after no more than $n$ rounds). 
\end{itemize}

The first claim follows immediately due to the definition of $t^*$. 
For the second claim, note that the total value of both classes after $t$ rounds must be at most $2t$ since only $2t$ items have arrived. Also, the total value of both classes after $t$ rounds must be at least $2(t-1)$; this is because the $t$th round only happens if some agent in $R^{t-1}_1 \cup R^{t-1}_2$ was not fully saturated after $t-1$ rounds, and since this agent was part of $R^{t'}_1 \cup R^{t'}_2$ for all $t' \le t-1$, non-wastefulness implies that the algorithm must have assigned the $2(t-1)$ items from the first $t-1$ rounds fully. These two claims, along with the convention that $V^{t^*}(N_1) \le V^{t^*}(N_2)$ implies the second claim. 

Before we prove the third claim, we show why these claims imply the desired bound on the envy ratio. 
At the end of the algorithm, the total value of class $N_1$ is at most $t^*$ because of the second claim and the fact that they do not receive any items from the last $2(n-t^*)$ items (as all agents in $R^{t^*}_1$ are saturated after round $t^*$). 

In contrast, the proportional fair share $\prop^S_1$ of class $N_1$ with respect to the allocated items $S$ is at least $n-1$. 
Note that all the items except for $o_{t^*}$ and $o'_{t^*}$ are fully assigned. Thus, $M\setminus \{o_{t^*},o'_{t^*}\} \subseteq S$. Further, consider two sets $P_1=\{o_1,\ldots,o_{t^*-1},o_{t*+1},\ldots,o_n\}$ and $P_2=\{o'_1,\ldots,o'_{t^*-1},o'_{t*+1},\ldots,o'_n\}$. From $P_1$, the $t^*-1$ items $o_1,o_2,\ldots,o_{t*-1}$ can be matched to $t^*-1$ agents $a_1,a_2,\ldots,a_{t*-1}$ and the remaining $n-t^*-$ items can be matched to $n-t^*$ agents in $R^{t^*}_1$. Similarly, from $P_2$, $t^*-1$ items $o'_1,o'_2,\ldots,o'_{t*-1}$ can be matched to $t^*$ agents $a'_1,a'_2,\ldots,a'_{t*-1}$ and the remaining $n-t^*$ items can be matched to $n-t^*$ agents in $R^{t^*}_1$. Thus, $\prop^S_1 \geq n-1$. From the third claim, if $V_1(X)  \geq \alpha \cdot \prop^S_1$, then  $(1-1/e)n  \geq \alpha (n-1)$, meaning that $(1-1/e) \frac{n}{n-1} \geq \alpha$.

Finally, we show that $t^* \le (1-1/e) \cdot n$. To see this, we first show that after $t$ rounds, 
\[
V^t(N_1 \setminus R^t_1) + V^t(N_2 \setminus R^t_2) \le \frac{2t}{n} + \frac{2(t-1)}{n-1} + \ldots + \frac{2\cdot 1}{n-t+1}.
\]

For the base case, note that after the first round, $V^1(a_1)+V^1(a'_1) \le 2/n$ follows from the pigeonhole principle. Suppose this claim holds after $t-1$ rounds. Then, after round $t$, we have 
\[
V^t(a_t)+V^t(a'_t) \le \frac{2t - (V^{t-1}(N_1 \setminus R^{t-1}_1)+V^{t-1}(N_2 \setminus R^{t-1}_2))}{n-t+1}.
\]
Adding $V^{t-1}(N_1 \setminus R^{t-1}_1)+V^{t-1}(N_2 \setminus R^{t-1}_2) = V^{t}(N_1 \setminus R^{t-1}_1)+V^{t}(N_2 \setminus R^{t-1}_2)$ to both sides, we obtain
\[
V^t(N_1 \setminus R^t_1) + V^t(N_2 \setminus R^t_2) \le \frac{2t}{n-t+1} + \frac{n-t}{n-t+1} \cdot (V^{t-1}(N_1 \setminus R^{t-1}_1)+V^{t-1}(N_2 \setminus R^{t-1}_2)).
\]

Using the induction hypothesis, we get the desired result. 
Consider the smallest $\hat{t}$ such that
\[
2\hat{t} - 2 - \left(\frac{2\hat{t}}{n} + \frac{2(\hat{t}-1)}{n-1} + \ldots + \frac{2\cdot 1}{n-\hat{t}+1}\right) \ge 2(n-\hat{t}).
\]
Note that the process must stop at $t^* \le \hat{t}$. This is because the total value of both classes after $\hat{t}$ round is at least $2\hat{t}-2$, but the value to the removed agents is at most the expression in the brackets. Hence, the remaining allocation must have saturated the remaining $2(n-\hat{t})$ agents. 
After simple algebra, we can see that the left hand side is equal to $2 \cdot (n-\hat{t}) \cdot (H_n-H_{n-\hat{t}}) - 2$. If this is at least $2(n-\hat{t})$, then $H_n-H_{n-\hat{t}} \ge 1+1/(n-\hat{t})$. The smallest $\hat{t}$ when this is satisfied is roughly $(1-1/e) \cdot n+o(n)$. 
\end{proof}

\subsection{Upper bounds for Algorithm~\ref{alg:divisible}}
\divisibleUpperBound*

\begin{proof}
The fact that Algorithm~\ref{alg:divisible} cannot achieve \CPROP{\alpha} for $\alpha > \e$ immediately follows from Theorem~\ref{thm:divupperbound}.

For each of the fairness or efficiency guarantees, we provide an instance for which Algorithm \ref{alg:divisible} cannot achieve the corresponding bound. 

\paragraph{\cef{}}
Consider the following instance with two classes $N_1=\{a_1,a_2,\ldots,a_n\}$ and $N_2=\{a'_1,a'_2,\ldots,a'_{2n}\}$. There are $2n$ items $o_1,o'_1,o_2,o'_2, \ldots,o_n,o'_n$. There are $n$ rounds: in round $t \in [n]$, item $o_t$ arrives, followed immediately by item $o'_t$. Each agent $a'_i$ $(i \in [2n])$ likes every item. Each agent $a_i$ $(i \in [n])$ likes the items $o_t$ and $o'_t$ with $t=1,2,\ldots,i$; namely, agent $a_1$ likes the items $o_1,o'_1$, agent $a_2$ likes items $o_1,o'_1,o_2,o'_2$, and so on. 

Note that since $N_2$ has $2n$ agents who like all $2n$ items, for each item, there is at least one agent in $N_2$ who is not saturated and likes that item. Thus, until the agents in $N_1$ who like $o_t$ and $o'_t$ are fully saturated, the equal-filling algorithm splits the item into halves among the two classes. The algorithm assigns the amount $\frac{1}{2n}$ of $\{o_t,o'_t\}$ to each agent in $N_2$. On the other hand, it assigns the amount $\frac{1}{n-(t-1)}$ of $o_t$ and $o'_t$ to each agent $i$ of class $N_1$ with $i \geq j$; for example, agent $a_1$ receives $\frac{1}{n}$ of $\{o_1,o'_1\}$; agent $a_2$ receives $\frac{1}{n}$ of $\{o_1,o'_1\}$ and $\frac{1}{n-1}$ of $\{o_2,o'_2\}$; agent $a_3$ receives $\frac{1}{n}$ of $\{o_1,o'_1\}$, $\frac{1}{n-1}$ of $\{o_2,o'_2\}$, and $\frac{1}{n-2}$ of $\{o_3,o'_3\}$; and so on. 

Let $X$ denote the matching returned by Algorithm~\ref{alg:divisible}. We will establish that $V_1(X) \leq (1-1/e) V^*_1(Y_2)$. 
First, it is not difficult to see that under $X$, class $N_2$ is assigned to at least $1$ for each item set of $\{o_t,o'_t\}$ ($t \in [n]$). Thus, $V^*_1(Y_2) \geq n$. Now, let $t^*=n-\lceil \frac{n}{e}\rceil$. 
It can be easily checked by the integral test that $\sum^{t^*}_{t=1}\frac{1}{n+1-t}$ is between $1-\frac{5}{n}$ and $1$. Thus, after the algorithm assigns $o_{t^*+5},o'_{t^*+5}$, the set $N_{1,o_t}$ becomes empty, i.e., there is no agent in $N_1$ who is not saturated and likes new items $o_{t},o'_{t}$ for $t > t^*+5$.
Thus, the value $V_1(X)$ derived by class $N_1$ from $X$ is at most 
\[
t^*+5 < (1-\frac{1}{e})n +5 \leq (1-\frac{1}{e})V^*_1(Y_2) +5,
\]
which proves the claim. 

\paragraph{\USW{}}
Let $n$ be a positive integer. Consider $n+1$ classes: There are $n$ classes $N_j$, each of which consists of a single agent $c_j$ for $j=1,2,\ldots,n$. The last class $N_{n+1}$ consists of $n$ agents $\{a_1,a_2,\ldots,a_n\}$. There are $2n$ items: $n$ red items $r_1,r_2,\ldots,r_n$ and $n$ blue items $b_1,b_2,\ldots,b_{n}$. Each red item is liked by every agent. Each blue item $b_i$ is liked by the single agent $c_{i}$ in $N_{i}$. Now the instance admits a perfect matching of size $2n$ that matches every agent $c_i$ for $i \in [n]$ to the blue item $b_{i}$ and the remaining $n$ agents in $N_{n+1}$ arbitrarily to the remaining $n$ red items. 

Now suppose that the items arrive in the order of $r_1,r_2,\ldots,r_n, b_1,b_2,\ldots,b_{n}$.
For each red item $r_i$ $(i \in [n])$, the equal-filling algorithm assigns an equal amount $\frac{1}{n+1}$ of fractions among the $n+1$ classes. Thus, after the algorithms matches the last red item $r_n$, the total amount of fractions each class $N_i$ for $i \in [n+1]$ has received is $\frac{n}{n+1}$. For each blue item $b_i$ $(i \in [n])$, the equal-filling algorithm assigns an amount of $\frac{1}{n+1}$ to the agent $c_i$ in $N_i$ since $c_i$ is the only agent who likes the blue item $b_i$ but has already been saturated up to $\frac{n}{n+1}$. 
Thus, the utilitarian social welfare of the resulting matching $X$ is given as follows: 
\[
\sum^n_{i=1}V_i(X)+V_{n+1}(X)=\sum^n_{i=1}1+ \frac{n}{n+1} = n+\frac{n}{n+1}. 
\]
This proves the claim. 
\end{proof}

\section{Omitted Material from Section~\ref{sec:randomized}} \label{app:randomized}

\subsection{Discussion on Other Randomized Algorithms}

Readers familiar with the online matching literature may wonder why can't we use the Ranking algorithm of \citet{karp1990optimal} to decide how to match items \emph{within} each class, and combine it with some fair class-level matching approach.
While we believe this is an interesting direction for future research, there is a concrete technical difficulty in analyzing such algorithms.
Naturally, the class-level matching must take into account which agents are already matched to previous items.
This means that the realization of randomness used by Ranking within some class $i$ will influence what items are allocated to the class!

How about applying Ranking directly, ignoring how agents are partitioned into classes?
While this approach circumvents the above challenge, it fails on two classes with lopsided sizes.
In the extreme, consider a class with only one agent, and another class with $n \gg 1$ agents, and only one item.
The second class will get the item with probability $\frac{n}{n+1}$ while the first class gets it only with probability $\frac{1}{n+1}$.

Finally, we observe that it is necessary to have randomness in both the class-level matching and the individual-level matching, in order to exploit the power of randomized algorithms.

\begin{proposition}
    If an algorithm assigns deterministically at the class-level, it is at best \CPROP{\frac{1}{2}}.
\end{proposition}

\begin{proof}
    Consider two classes $N_1 = \{a_1, a_2, a_3\}$ and $N_2 = \{b_1, b_2, b_3\}$.
    For $1 \le i \le 3$, the $i$-th item is liked by $a_i$ and $b_i$.
    If the algorithm assigns all three items to the same class, it is only \CPROP{0}.
    Otherwise, assume without loss of generality that $2$ items go to class $2$.
    Let the next item be only liked by the matched agent in class $1$ and the unmatched agent in class $2$, as in \Cref{fig::GEF_non}.
    The algorithm is then at best \CPROP{\frac{1}{2}}.
\end{proof}

\begin{proposition}
    If an algorithm assigns deterministically within each class, it is at best \CPROP{\frac{1}{2}}.
\end{proposition}

\begin{proof}
    It becomes apparent when we consider a single class.
    The proposition then reduces to the fact that deterministic online matching algorithms are at best $\frac{1}{2}$-competitive.
    We can extend this hard instance to $k$ classes by making $k$ copies of the class and each item.
\end{proof}

\subsection{Discussion on Randomized Algorithms and CEF}

As discussed in the last subsection, if the the class-level matching depends on which agents are already matched, i.e., if it is adaptive to the realization of randomness in the agent-level matching, then the realization of randomness in an online algorithm, e.g., Ranking, within each class would affect what items get assigned to the class.
How about using a class-level matching algorithm that is oblivious to the randomness in the agent-level matching?
Although such algorithms must violate non-wastefulness in general, we find an algorithm that isn't blatantly wasteful and looks interesting enough to be a stepping stone towards stronger algorithms in future works.

We call this algorithm \algdiscussion.
For each item, it randomly assigns the item to a class with at least one agent who likes the item.
Within each class, it runs a separate Ranking algorithm to match items to agents therein.

\begin{proposition}
    Given an online indivisible instance, \algdiscussion guarantees \cef{(\e)}.
\end{proposition}

\begin{proof}
    Consider any class $i$ and any other class $j$.
    Let $y_i = (y_{io})_{o \in M} \in \{0, 1\}^M$ be the vector that represents the subset of items assigned to $i$ by \algdiscussion at the class-level, regardless of whether such the items are matched to agents successfully.
    Define $y_j$ similarly.
    Note that both $y_i$ and $y_j$ are random variables that depend on the class-level random assignments of items.
    Finally, let $X = (x_{ao})_{a \in N, o \in M} \in \{0, 1\}^{N \times M}$ be the matrix that represents the matching by \algdiscussion.
    We seek to prove that:
    \[
        \E[V_i(X)] \ge (\e)\,\E [V_i^*(y_j)] 
        ~.
    \]
    
    Conditioned on the subset of items assigned to $i$, i.e., $y_i$, the Ranking algorithm ensures that (see, e.g., \citet{karp1990optimal}):
    \[
        \E\,\big[ V_i(X) \mid y_i \big] \ge (\e)\,V_i^*(y_i)
        ~.
    \]
    
    It remains to show that:
    \[
        \E [V_i^*(y_i)] \ge \E [V_i^*(y_j)] 
        ~.
    \]
    
    Define $\hat{y}_j$ be such that $\hat{y}_{jo} = y_{jo}$ if class $i$ has at least one agent who likes item $o$, and $\hat{y}_{jo} = 0$ otherwise.
    By definition $V_i^*(y_j) = V_i^*(\hat{y}_j)$ and therefore it suffices to prove:
    \[
        \E [V_i^*(y_i)] \ge \E [V_i^*(\hat{y}_j)] 
        ~.
    \]
    
    Note that for any item $o$, \algdiscussion ensures that the probability that $y_{io} = 1$ is greater than or equal to the probability that $\hat{y}_{jo} = 0$.
    Further, the assignment of items at the class-level are independent.
    Hence we get that random variable $y_i$ stochastically dominates $\hat{y}_j$.
    The above inequality now follows by the monotonicity of $V_i^*$.
\end{proof}

\end{document}